\documentclass[reqno,11pt]{amsart}

\UseRawInputEncoding

\newtheorem{corollary}{Corollary}
\newtheorem{definition}{Definition}

\usepackage{booktabs}

\usepackage{graphicx}
\usepackage{xcolor}
\usepackage{tikz}

\usepackage{tikz-3dplot}
\usepackage{xifthen}
\usepackage{caption}

\newcommand{\be}{\begin{equation}}
\newcommand{\ee}{\end{equation}}
\newcommand{\dis}{\displaystyle}

\usepackage{enumerate}
\usepackage{enumitem}

\usepackage{hyperref}

\newtheorem{thm}{Theorem}[section]
\newtheorem{prop}[thm]{Proposition}

\newtheorem{lemma}[thm]{Lemma}

\newtheorem{remark}[thm]{\it Remark}
\newtheorem{example}[thm]{\it Example}

\title{ Discrete Lax pairs and   hierarchies of integrable difference systems}

\author{Pavlos Kassotakis }
\email{pavlos1978@gmail.com}

\begin{document}
\maketitle

\begin{abstract}
We introduce a family of order $N\in \mathbb{N}$ Lax matrices that   is indexed by the natural number $k\in \{1,\ldots,N-1\}.$ For each value of   $k$ they serve as strong Lax matrices of a hierarchy of integrable difference systems in edge variables that in turn lead to hierarchies of integrable difference systems in vertex variables or in a combination of edge and vertex variables.   Furthermore, the entries of the Lax matrices are considered as elements of  a division ring, so we obtain  hierarchies of discrete integrable systems extended in the non-commutative  domain.
\end{abstract}

\vspace{2pc}
\noindent{\it Keywords}: Discrete Lax pairs, hierarchies, non-commutative integrable difference systems, Yang-Baxter maps

\setcounter{tocdepth}{2}

\tableofcontents


\section{Introduction}
Multi-component versions of scalar  integrable evolution equations were introduced  by the Russian school during  1970's. Gel'fand and Dikii \cite{Gel:1976}, Manin \cite{Manin:1979}, as well as Drinfeld and Sokolov \cite{Drinfeld:1981}, by studying higher-order spectral problems,  which generalize spectral problems of KdV type, they arrived to multi-component hierarchies of integrable evolution equations of KdV type. These studies on higher order-spectral problems were continued  by the seminal contributions of Mikhailov \cite{Mikhailov:1979}, Fordy et.al \cite{Fordy:1980,Fordy:1993}  and recently by Adler and Sokolov \cite{Adler:2021} in the non-commutative setting.

In the discrete scheme,
results on multi-component partial difference systems, or scalar partial difference equations defined on higher order stencils, are rather sparse \cite{Nijhoff:1996,Tongas:2004,Maruno:2010,Hietarinta:2011,JamesPhd,Hay:2014,Mikhailov:2016,Nalini:2018,Kels:2019,Kels:2019II,Kass2,Kamp:2020,Hietarinta:2020}. Moreover, there are just a few results on hierarchies of integrable partial difference equations, for example discrete analogues of the  Gel'fand-Dikii hierarchy, of the modified and the Schwarzian Gel'fand-Dikii hierarchies have been introduced in the past. Specifically in \cite{Nijhoff:1992}, a hierarchy of discrete equations  which under suitable continuous limits leads to the Gel'fand-Dikii hierarchy was presented and  it was referred to as the  {\it the lattice Gel'fand-Dikii hierarchy}. The first two members of the  lattice Gel'fand-Dikii hierarchy respectively are the lattice potential KdV equation and the lattice version of the Boussinesq equation. Furthermore, the first two members of the lattice-modified  Gel'fand-Dikii hierarchy, namely the lattice-modified KdV and Boussinesq equation  were also explicitly presented in \cite{Nijhoff:1992}. The full explicit presentation of all members of the lattice-modified  Gel'fand-Dikii hierarchy had to wait for \cite{JamesPhd} and \cite{Atkinson:2012}. Moreover, the first two members of the lattice-Schwarzian Gel'fand-Dikki hierarchy i.e. the lattice-Schwarzian KdV and lattice-Schwarzian Boussinesq equations,  were firstly introduced in \cite{Ni1} and the whole hierarchy was presented in \cite{JamesPhd}. In addition, an extension of the lattice-modified  Gel'fand-Dikii hierarchy to the non-commutative domain was considered in \cite{Doliwa_2013,Doliwa_2014}.

The results of this paper serve as a contribution to the growing interest of deriving and extending integrable difference systems to the non-commutative setting \cite{Nijhoff:1990,Boris:2000,Dimakis:2002,bs:2002N,Field:2005,Nimmo:2006,Doliwa_2013,Doliwa_Painleve_2013,Doliwa_2014,Grahovski:2016,Rizos:2016,Rizos:20182,Kass1,Noumi:2020}. Specifically, by introducing a family of discrete Lax matrices of order $N\in\mathbb{N}$ that we denote as $L^{N,k}$ (see Section \ref{secdef:lax}), with entries  elements of a division ring, for each value of the index $k\in\{1,\ldots,N-1\}$, we obtain a hierarchy of difference systems in non-commutative edge variables that in turn leads to hierarchies of difference systems in non-commutative vertex variables or in a combination of edge and vertex variables. Hierarchies of integrable difference systems that correspond to a Lax matrix with a specific index $k,$ can as well arise as reductions of hierarchies which correspond to a Lax matrix with index $k'>k$. In that respect, in this paper we present a hierarchy of hierarchies of integrable difference systems in edge and vertex variables.

The outline of this paper is as follows. In Section \ref{Section:2}, after introducing the notation and definitions used throughout this paper, we introduce the Lax matrices $L^{N,k}$ and the discrete spectral problem that they participate. Moreover, we prove that these Lax matrices are strong (see Definition \ref{Lax2:def1}) and we provide implicitly the associated family of hierarchies of difference systems in non-commutative edge variables. In Section \ref{Section:3}, we derive explicitly a hierarchy of difference systems in edge non-commutative variables associated with the Lax matrix $L^{N,1}$ for arbitrary $N$ and we prove integrability. Furthermore, we obtain the associated hierarchies of difference systems in vertex variables. We show that when certain centrality assumptions are imposed, the explicit form of $2N-$parameter extensions of the non-commutative lattice-modified and lattice-Schwarzian Gel'fand-Dikki hierarchies are obtained. Also, we provide explicitly  the first two members of the hierarchy of Yang-Baxter maps that correspond to these hierarchies and implicitly the full hierarchy of Yang-Baxter maps. In Section~\ref{Section:4} we obtain the explicit form of a hierarchy of difference systems in edge non-commutative variables associated with the Lax matrix $L^{N,2},$ that results a hierarchy in vertex variables. Furthermore, we show that this hierarchy includes both hierarchies, obtained by the linear problem associated with  $L^{N,1}$, as reductions. 
We end this paper with  Section \ref{Section:5}, where conclusions and perspectives for future research are presented.

\section{Notation, definitions and the family of Lax matrices $L^{N,k}$} \label{Section:2}
The $\mathbb{Z}^2$ graph is defined as the graph with set of vertices $V=\left\{(m,n)  |   m,n \in {\mathbb Z} \right\}$ and set of edges \hbox{$E=E_H \sqcup E_V $}, i.e. the disjoint union of
horizontal edges $E_H= \left\{ \{(m,n),(m+1,n)\}  |  m,n \in {\mathbb Z} \right\}$ and the  vertical ones $E_V=\left\{ \{(m,n),(m,n+1)\}  |   m,n \in {\mathbb Z} \right\}$ (see Figure \ref{fig00}).
\begin{figure}[htb]
\begin{center}
\begin{minipage}[h]{0.45\textwidth}\resizebox{1.0\textwidth}{!}{ \begin{tikzpicture}[>=stealth',node distance=1.5cm, thick,main node/.style={draw,circle,inner sep=1.5pt,minimum size=2.5pt,fill=black,font=\bfseries},shorten > = 1pt,
node distance = 3cm and 4cm,el/.style = {inner sep=2pt, align=left, sloped},
every label/.append style = {font=\tiny}]
  \node[main node] (1) {}; \node[main node] (11) [right of=1] {}; \node[main node] (111) [right of=11] {};
  \node[main node] (2) [below  of=1] {};  \node[main node] (21) [right of=2] {};  \node[main node] (211) [right of=21] {};
  \node[main node] (3) [above of=1] {};   \node[main node] (31) [right of=3] {};   \node[main node] (311) [right of=31] {};

  \node (31l) [left of=11, node distance=0.9cm,label={(m+1,n+1)}] {};
     \node (21l) [left of=21, node distance=0.9cm,label={(m+1,n)}] {};

      \node (-2) [below of=2, node distance=0.9cm] {};       \node (-3) [above of=3,node distance=0.9cm] {};
      \node (-21) [below of=21, node distance=0.9cm] {};     \node (3l) [left of=3,node distance=0.9cm] {};
       \node (-31) [above of=31,node distance=0.9cm] {};
  \node (-211) [below of=211, node distance=0.9cm] {};       \node (-311) [above of=311,node distance=0.9cm] {};
             \node (2l) [left of=2, node distance=0.9cm,label={(m,n)}] {};
             \node (1l) [left of=1,node distance=0.9cm,label={(m,n+1)}] {};             \node (2r) [right of=211, node distance=0.9cm] {};
             \node (3r) [right of=311,node distance=0.9cm] {}; \node (1r) [right of=111,node distance=0.9cm] {};
             \node (211l) [left of=211, node distance=0.9cm,label={(m+2,n)}] {};
                          \node (3l) [left of=3, node distance=0.9cm,label={(m,n+2)}] {};
               \path
  (-2) edge[dotted] (2) edge[dotted] (1) edge[dotted] (3) edge[dotted] (-3)
  (-21) edge[dotted]    (21) edge[dotted] (11)  edge[dotted] (31) edge[dotted] (-31)
  (-211) edge[dotted]      (211) edge[dotted] (111) edge[dotted] (311) edge[dotted] (-311)
(1l) edge[dotted]    (1) edge[dotted] (11) edge[dotted] (111) edge[dotted] (1r)
(2l) edge[dotted]    (2) edge[dotted] (21) edge[dotted] (211)  edge[dotted] (2r)
(3l) edge[dotted]    (3) edge[dotted] (31) edge[dotted] (311)  edge[dotted] (3r)
  ;
  \end{tikzpicture}}
 \captionsetup{font=footnotesize}
\captionof*{figure}{(a)}
\end{minipage}
\begin{minipage}[h]{0.42\textwidth}
 \resizebox{1.0\textwidth}{!}{ \begin{tikzpicture}[>=stealth',node distance=1.5cm, thick,main node/.style={draw,circle,inner sep=0,minimum size=0.5pt,fill=black,font=\bfseries},shorten > = 1pt,
node distance = 3cm and 4cm,el/.style = {inner sep=2pt, align=left, sloped},
every label/.append style = {font=\tiny}]
  \node[main node] (1) {};
  \node[main node] (2) [below  of=1] {};
  \node[main node] (3) [above of=1] {};
  \node[main node] (11) [right of=1] {}; \node[main node] (111) [right of=11] {};
 \node[main node] (21) [right of=2] {};  \node[main node] (211) [right of=21] {};
  \node[main node] (31) [right of=3] {};   \node[main node] (311) [right of=31] {};
      \node (-2) [below of=2, node distance=0.9cm] {};       \node (-3) [above of=3,node distance=0.9cm] {};
      \node (-21) [below of=21, node distance=0.9cm] {};       \node (-31) [above of=31,node distance=0.9cm] {};
      \node (-211) [below of=211, node distance=0.9cm] {};       \node (-311) [above of=311,node distance=0.9cm] {};
             \node (2l) [left of=2, node distance=0.9cm] {};       \node (3l) [left of=3,node distance=0.9cm] {}; \node (1l) [left of=1,node distance=0.9cm] {};
             \node (2r) [right of=211, node distance=0.9cm] {};       \node (3r) [right of=311,node distance=0.9cm] {}; \node (1r) [right of=111,node distance=0.9cm] {};
  \path
  (-2) edge[dashed] (2) edge[dashed] (1) edge[dashed] (3)   edge[dashed] (-3)
  (-21) edge[dashed]    (21) edge[dashed] (11)  edge[dashed] (31) edge[dashed] (-31)
  (-211) edge[dashed]      (211) edge[dashed] (111) edge[dashed] (311) edge[dashed] (-311)
(1l) edge    (1) edge (11) edge (111) edge (1r)
(2l) edge    (2) edge (21) edge (211)  edge (2r)
(3l) edge    (3) edge[dotted] (31) edge (311)  edge (3r)  ;
  \draw  (2) -- node [scale=0.6,below=2pt,fill=white]  { $ (m+1/2,n)$} (21); \draw  (21) -- node [scale=0.6,below=2pt,fill=white]  { $ (m+3/2,n)$} (211);
  \draw  (1) -- node [scale=0.6,below=2pt,fill=white]  { $ (m+1/2,n+1)$} (11);
  \draw[dotted]  (2)  -- node [scale=0.6,rotate=90,above=4pt,fill=white]  { $ (m,n+1/2)$} (1);\draw[dotted]  (1)  -- node [scale=0.6,rotate=90,above=4pt,fill=white]  { $ (m,n+3/2)$} (3);
    \draw[dotted]  (21)  -- node [scale=0.6,rotate=90,above=4pt,fill=white]  { $ (m+1,n+1/2)$} (11);
   \end{tikzpicture}}
\captionsetup{font=footnotesize}
\captionof*{figure}{(b)}
\end{minipage}
\caption{(a): The set $V$ where the dependent variables of the vertex equations are assigned. (b): The sets of horizontal $E_H$ (solid lines) and the set of vertical edges $E_V$ (dashed lines) where the dependent variables  of the edge equations are assigned.}
\label{fig00}
\end{center}
\end{figure}
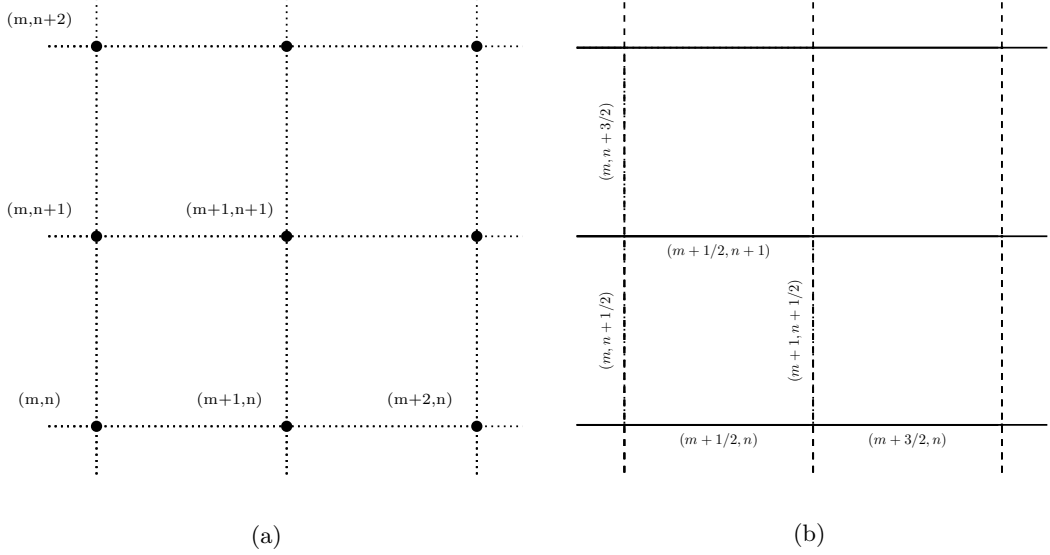
 It is convenient to label with $(m,n) \in {\mathbb Z}^2,$
$(m+1/2,n) \in {\mathbb Z}^2$ and $(m,n+1/2) \in {\mathbb Z}^2$   the elements of $V,$ $E_H$ and $E_V,$  respectively.

We consider the functions $\phi^i: V\ni (m,n) \mapsto \phi^{i}_{m,n}\in \mathbb{D},$   $x^i: E_H\ni (m,n+1/2) \mapsto x^{i}_{m,n+1/2}\in \mathbb{D}$ and  $y^i: E_V\ni (m+1/2,n)\mapsto y^{i}_{m+1/2,n} \in \mathbb{D},$  $i=1,2,\ldots,N,$ where $\mathbb{D}$ a non-commutative division ring f.i. an associative algebra over the field of complex numbers $\mathbb{C}$ with a multiplicative identity element denoted by $1$ and  every non-zero element $x\in \mathbb{D}$ has a unique multiplicative inverse denoted by $x^{-1}$ s.t. $xx^{-1}=x^{-1}x=1$.  We also  consider the functions $p^i: E_h\ni(m+1/2,n)\mapsto p^i_{m+1/2}\in\mathbb{D}$ and $q^i: E_v\ni(m,n+1/2)\mapsto q^i_{n+1/2}\in\mathbb{D}$ $i=1,2,\ldots,N,$ that we assume that they are elements of the center of the algebra $\mathbb{D}$ i.e. they commute with every element of $\mathbb{D}$. In this article, the functions $p^i, q^i$ will be simply referred to as {\it parameters}.
We simplify the notation above by denoting $ \phi^i:=\phi^i_{m,n},$ $\phi^i_2:=\phi^i_{m,n+1},$ $\phi^i_1:=\phi^i_{m+1,n}$, $\phi^i_{12}:=\phi^i_{m+1,n+1},$ $i=1,2,\ldots,N$ etc. By slightly abusing this notation, we also denote $ x^i:=x^i_{m+1/2,n},$ $ x^i_2:=x^i_{m+1/2,n+1},$ $ y^i_1:=y^i_{m+1,n+1/2},$ $y^i_{2}:=y^i_{m,n+3/2},$ $i=1,2,\ldots,N$ etc. (see Figure \ref{notation1}).
\begin{figure}[htb]
\begin{minipage}[h]{0.45\textwidth}
 \resizebox{1.1\textwidth}{!}{ \begin{tikzpicture}[fill=white,>=stealth',node distance=1.5cm, thick,main node/.style={fill=white,font=\bfseries},shorten > = 1pt,
node distance = 3cm and 4cm,el/.style = {inner sep=2pt, align=left, sloped},
every label/.append style = {font=\tiny},fill=white, scale=1.8, every node/.style={transform shape}]
  \node[main node,fill=white] (1) {};
  \node[main node] (2) [below  of=1] {};
  \node[main node] (3) [above of=1] {};
  \node[main node] (11) [right of=1] {}; \node[main node] (111) [right of=11] {};
 \node[main node] (21) [right of=2] {};  \node[main node] (211) [right of=21] {};
  \node[main node] (31) [right of=3] {};   \node[main node] (311) [right of=31] {};
      \node (-2) [below of=2, node distance=0.9cm] {};       \node (-3) [above of=3,node distance=0.9cm] {};
      \node (-21) [below of=21, node distance=0.9cm] {};       \node (-31) [above of=31,node distance=0.9cm] {};
      \node (-211) [below of=211, node distance=0.9cm] {};       \node (-311) [above of=311,node distance=0.9cm] {};
             \node (2l) [left of=2, node distance=0.9cm] {};       \node (3l) [left of=3,node distance=0.9cm] {}; \node (1l) [left of=1,node distance=0.9cm] {};
             \node (2r) [right of=211, node distance=0.9cm] {};       \node (3r) [right of=311,node distance=0.9cm] {}; \node (1r) [right of=111,node
             distance=0.9cm] {};
  \draw  (2) -- node [scale=1,below=2pt,fill=white]  { $ x^i_{m+1/2,n}$} node [scale=1,above=2pt,fill=white]  { $ p^i_{m+1/2}$} (21);
  \draw  (1) -- node [scale=1,below=2pt,fill=white]  { $ x^i_{m+1/2,n+1}$} node [scale=1,above=2pt,fill=white]  { $ p^i_{m+1/2}$} (11);
  \draw  (2)   --  node [scale=1,rotate=90,above=4pt,fill=white]  { $ y^i_{m,n+1/2}$} node [scale=1,below=0pt,rotate=90,fill=white]  { $ q^i_{n+1/2}$} (1);
    \draw[dotted]  (21)  --  node[scale=1,left=1pt,rotate=90,above=4pt,fill=white]  { $ y^i_{m+1,n+1/2}$} node [scale=1,below=0pt,rotate=90,fill=white]  { $ q^i_{n+1/2}$} (11);
      \draw  (21) -- node [scale=1,below=2pt,fill=white]  { $ x^i_{m+3/2,n}$} node [scale=1,above=2pt,fill=white]  { $ p^i_{m+3/2}$} (211);
   \draw[dotted]  (1)  --  node[scale=1,left=4pt,rotate=90,above=4pt,fill=white]  { $ y^i_{m,n+3/2}$} node [scale=1,below=0pt,rotate=90,fill=white]  { $ q^i_{n+3/2}$} (3);
  \path
  (-2) edge[dashed,fill=white] (2) edge[dashed,fill=white] (1) edge[dashed,fill=white] (3)   edge[dashed,fill=white] (-3)
  (-21) edge[dashed,fill=white]    (21) edge[dashed,fill=white] (11)  edge[dashed,fill=white] (31) edge[dashed,fill=white] (-31)
  (-211) edge[dashed,fill=white]      (211) edge[dashed,fill=white] (111) edge[dashed,fill=white] (311) edge[dashed,fill=white] (-311)
(1l) edge    (1) edge (11) edge (111) edge (1r)
(2l) edge    (2) edge (21) edge (211)  edge (2r)
(3l) edge    (3) edge[dotted] (31) edge (311)  edge (3r)  ;
  \draw (1) node [scale=1,fill=white] {$\phi^i_{m,n+1}$};  \draw (11) node [scale=1,fill=white] {$\phi^i_{m+1,n+1}$};  \draw (2) node [scale=1,fill=white] {$\phi^i_{m,n}$}; \draw (21) node [scale=1,fill=white] {$\phi^i_{m+1,n}$};\draw (211) node [scale=1,fill=white] {$\phi^i_{m+2,n}$}; \draw (3) node [scale=1,fill=white] {$\phi^i_{m,n+2}$};
   \end{tikzpicture}}
\captionsetup{font=footnotesize}
\captionof*{figure}{(a)}
\end{minipage}\;\;\;\;\;\;
\begin{minipage}[h]{0.45\textwidth}
 \resizebox{1.1\textwidth}{!}{ \begin{tikzpicture}[fill=white,>=stealth',node distance=1.5cm, thick,main node/.style={fill=white,font=\bfseries},shorten > = 1pt,
node distance = 3cm and 4cm,el/.style = {inner sep=2pt, align=left, sloped},
every label/.append style = {font=\tiny},fill=white, scale=1.8, every node/.style={transform shape}]
  \node[main node,fill=white] (1) {};
  \node[main node] (2) [below  of=1] {};
  \node[main node] (3) [above of=1] {};
  \node[main node] (11) [right of=1] {}; \node[main node] (111) [right of=11] {};
 \node[main node] (21) [right of=2] {};  \node[main node] (211) [right of=21] {};
  \node[main node] (31) [right of=3] {};   \node[main node] (311) [right of=31] {};
      \node (-2) [below of=2, node distance=0.9cm] {};       \node (-3) [above of=3,node distance=0.9cm] {};
      \node (-21) [below of=21, node distance=0.9cm] {};       \node (-31) [above of=31,node distance=0.9cm] {};
      \node (-211) [below of=211, node distance=0.9cm] {};       \node (-311) [above of=311,node distance=0.9cm] {};
             \node (2l) [left of=2, node distance=0.9cm] {};       \node (3l) [left of=3,node distance=0.9cm] {}; \node (1l) [left of=1,node distance=0.9cm] {};
             \node (2r) [right of=211, node distance=0.9cm] {};       \node (3r) [right of=311,node distance=0.9cm] {}; \node (1r) [right of=111,node
             distance=0.9cm] {};
  \draw  (2) -- node [scale=1,below=2pt,fill=white]  { $ x^i$} node [scale=1,above=2pt,fill=white]  { $ p^i$} (21);
  \draw  (1) -- node [scale=1,below=2pt,fill=white]  { $ x^i_2$} node [scale=1,above=2pt,fill=white]  { $ p^i$} (11);
  \draw  (2)   --  node [scale=1,rotate=90,above=4pt,fill=white]  { $ y^i$} node [scale=1,below=0pt,rotate=90,fill=white]  { $ q^i$} (1);
    \draw[dotted]  (21)  --  node[scale=1,left=10pt,rotate=90,above=4pt,fill=white]  { $ y^i_1$} node [scale=1,below=0pt,rotate=90,fill=white]  { $ q^i$} (11);
      \draw  (21) -- node [scale=1,below=2pt,fill=white]  { $ x^i_1$} node [scale=1,above=2pt,fill=white]  { $ p^i_1$} (211);
   \draw[dotted]  (1)  --  node[scale=1,left=4pt,rotate=90,above=4pt,fill=white]  { $ y^i_{2}$} node [scale=1,below=0pt,rotate=90,fill=white]  { $ q^i_2$} (3);
  \path
  (-2) edge[dashed,fill=white] (2) edge[dashed,fill=white] (1) edge[dashed,fill=white] (3)   edge[dashed,fill=white] (-3)
  (-21) edge[dashed,fill=white]    (21) edge[dashed,fill=white] (11)  edge[dashed,fill=white] (31) edge[dashed,fill=white] (-31)
  (-211) edge[dashed,fill=white]      (211) edge[dashed,fill=white] (111) edge[dashed,fill=white] (311) edge[dashed,fill=white] (-311)
(1l) edge    (1) edge (11) edge (111) edge (1r)
(2l) edge    (2) edge (21) edge (211)  edge (2r)
(3l) edge    (3) edge[dotted] (31) edge (311)  edge (3r)  ;
  \draw (1) node [scale=1,fill=white] {$\phi^i_2$};  \draw (11) node [scale=1,fill=white] {$\phi^i_{12}$};  \draw (2) node [scale=1,fill=white] {$\phi^i$}; \draw (21) node [scale=1,fill=white] {$\phi^i_1$};\draw (211) node [scale=1,fill=white] {$\phi^i_{11}$}; \draw (3) node [scale=1,fill=white] {$\phi^i_{22}$};
   \end{tikzpicture}}
   \captionsetup{font=footnotesize}
\captionof*{figure}{(b)}
\end{minipage}
\caption{Dependent variables assigned on  the vertices and on the edges of the  ${\mathbb Z}^2$ graph.  (a): Standard notation. (b): Notation used in the paper}\label{notation1}
\end{figure}
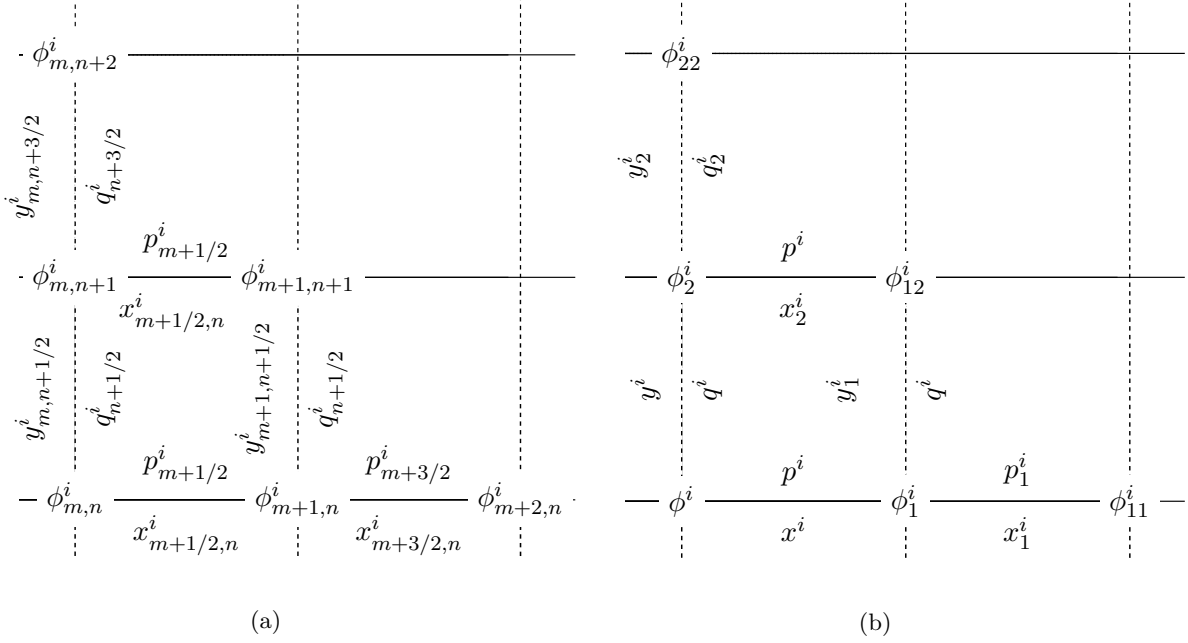

A set of equations involving  edge variables f.i. $x^i, y^i$ $i=1,2,\ldots, N$  and a finite number of their shifts, is called  {\em system of difference equations in edge variables}.
A set of equations involving   vertex variables f.i. $\phi^i, i=1,2,\ldots, N$   and a finite number of their shifts, is called  {\em system of difference equations in vertex variables}. Finally, a set of equations involving  edge and vertex variables f.i. $x^i, y^i$ $i=1,2,\ldots, N_1,$  $\phi^i, i=1,2,\ldots, N_2,$  and a finite number of their shifts, is called  {\em system of difference equations in edge and vertex variables}.

Let us denote with $X,$ respectively $Y,$ the sets $\{x^1,\ldots,x^N\},$ respectively $\{y^1,\ldots,y^N\},$ of dependent variables, as well as the sets of parameters $P:=\{p^1,\ldots, p^N\}$ and $Q:=\{q^1,\ldots, q^N\},$ $N\in\mathbb{N}.$ We now proceed to the following definitions.
\begin{definition} \label{Lax2:def1}
A matrix of order $N,$ $L(X;P,\lambda)$ is called a Lax matrix of  the difference system in edge variables
\begin{align} \label{eq-edge}
 x^i_2=F^i(X,Y;P,Q),\quad  y^i_1=G^i(X,Y;P,Q),\quad i=1,\ldots, N,\;\; N\in \mathbb{N}
\end{align}
where $F^i, G^i,\;\;i=1,\ldots ,N$ are  functions of the indicated variables,  if (\ref{eq-edge})
 implies that
\begin{align} \label{lax-eq}
L(X_2;P,\lambda)\, L(Y;Q,\lambda) =
L(Y_1;Q,\lambda)\, L(X;P,\lambda),
\end{align}
holds for all $\lambda$ where $\lambda$ the spectral parameter.
$L(X;P,\lambda)$ is called a strong Lax matrix of
(\ref{eq-edge}), if the converse also holds.
\end{definition}
The matrix equation (\ref{lax-eq}) is referred to as {\it the discrete Lax equation} or the discrete {\it zero-curvature condition} and arises as the compatibility condition of the following linear system:
\begin{align*}
 \Psi_2=L(X;P,\lambda) \Psi, \quad  \Psi_1=L(Y;Q,\lambda) \Psi,
\end{align*}
where $\Psi$ stands for an $N-$component vector.

\begin{definition} \label{def2}
The difference system in edge variables (\ref{eq-edge}) will be called {\em birational} if it implies
\begin{align} \label{eq-edge-b}
 x^i=f^i(X_2,Y_1;P,Q),\quad  y^i=g^i(X_2,Y_1;P,Q),\quad i=1,\ldots, N,\;\; N\in \mathbb{N},
\end{align}
where $f^i, g^i,\;\;i=1,\ldots ,N$ are  functions of the indicated variables. Moreover, a birational system f.i. (\ref{eq-edge}),  will be called {\em quadrirational} \cite{ABS:YB} if it implies that
\begin{align} \label{eq-edge-q}
 x^i=H^i(X_2,Y;P,Q),\quad  y^i_1=K^i(X_2,Y;P,Q),\quad i=1,\ldots, N,\;\; N\in \mathbb{N},
\end{align}
where $H^i, K^i,\;\;i=1,\ldots ,N$ are  functions of the indicated variables and (\ref{eq-edge-q}) is a birational system. Furthermore, system (\ref{eq-edge-q}) will be referred to as the {\em companion system} of the quadrirational system in edge variables (\ref{eq-edge}).
\end{definition}

\begin{definition}  \label{adm:Mob}
The following change of  dependent variables
\begin{align*}
x^i_2\mapsto  \left(a^i(P)x^i_2+b^i(P)\right)\left(c^i(P)x^i_2+d^i(P)\right)^{-1},&&x^i\mapsto\left(a^i(P)x^i+b^i(P)\right)\left(c^i(P)x^i+d^i(P)\right)^{-1},\\
y^i_1\mapsto  \left(a^i(Q)y^i_1+b^i(Q)\right)\left(c^i(P)y^i_1+d^i(P)\right)^{-1},&&y^i\mapsto\left(a^i(Q)y^i+b^i(Q)\right)\left(c^i(P)y^i+d^i(P)\right)^{-1},
\end{align*}
or
\begin{align*}
x^i_2\mapsto  \left(a^i(P)x^i_2+b^i(P)\right)^{-1}\left(c^i(P)x^i_2+d^i(P)\right),&&x^i\mapsto\left(a^i(P)x^i+b^i(P)\right)^{-1}\left(c^i(P)x^i+d^i(P)\right),\\
y^i_1\mapsto  \left(a^i(Q)y^i_1+b^i(Q)\right)^{-1}\left(c^i(P)y^i_1+d^i(P)\right),&&y^i\mapsto\left(a^i(Q)y^i+b^i(Q)\right)^{-1}\left(c^i(P)y^i+d^i(P)\right),
\end{align*}
$\forall i\in \{1,\ldots,N\},$ where $a^i,b^i,c^i$ and $d^i$   functions of the indicated parameters,
will be called {\it admissible Mobi\"us  transformations} or simply $(M\ddot{o}b)^2$ transformations 
 of the difference systems in edge variables (\ref{eq-edge}).
Two difference systems in edge variables which are related by an admissible Mobi\"us transformation will be called {\it equivalent}.
\end{definition}
Note that admissible Mobi\"us  transformations respect the {\it multidimensional compatibility} a.k.a integrability of  a difference system in edge variables \cite{Papageorgiou:2010}.

Lets denote with $\Phi$ the set $\{\phi^1,\ldots, \phi^M\},\; M\in\mathbb{N}$ of the dependent variables. We proceed to the following definition
\begin{definition}
A matrix of order $N,$ $L(\Phi_1,\Phi;P,\lambda)$ is called a Lax matrix of  the difference system in vertex variables
\begin{align} \label{eq-vertex}
 \phi^i_{12}=H^i(\Phi,\Phi_1,\Phi_2;P,Q),\quad i=1,\ldots, M,\;\; M\in \mathbb{N}
\end{align}
where $H^i,\;\;i=1,\ldots , M$ are  functions of the indicated variables,  if (\ref{eq-vertex})
 implies that
\begin{equation} \label{lax-eq-ver}
L(\Phi_{12}, \Phi_2;P,\lambda)\, L(\Phi_2,\Phi;Q,\lambda) =
L(\Phi_{12}, \Phi_1;Q,\lambda)\, L(\Phi_1,\Phi;P,\lambda),
\end{equation}
holds for all $\lambda$ where $\lambda$ the spectral parameter.
$L(\Phi_1,\Phi;\lambda)$ is called a strong Lax matrix of
$(\ref{eq-vertex}),$ if the converse also holds.
\end{definition}

\begin{definition}
The order $N$ lower-triangular nilpotent matrices $\nabla^k,$ $k=1,2,\ldots, N-1$ are defined by
$$
(\nabla^k)_{ij}:=\left\{ \begin{array}{ll}
            0,& i\leq j \\
            \delta_{i,j+k},& i>j
                \end{array} \right.
$$
and are said to have {\it level} $-k.$
\end{definition}

\begin{definition}
The order $N$ upper-triangular nilpotent matrices $\Delta^k,$ $k=1,2,\ldots, N-1$ are defined by
$$
(\Delta^k)_{ij}:=\left\{ \begin{array}{ll}
            \delta_{i+N-k,j},& i< j \\
              0,& i\geq j
                \end{array} \right.
                $$
and are said to have {\it level} $N-k.$
\end{definition}
\begin{remark} \label{rem1}
For $1\leq k,l \leq N-1$ it holds:
$$
\nabla^k \nabla^l= \left\{ \begin{array}{ll}
            \nabla^{k+l},& k+l\leq N-1 \\
              0,& k+l>N-1
                \end{array} \right.
, \quad \Delta^k \Delta^l= \left\{ \begin{array}{ll}
            \Delta^{k+l}\; (mod \;\;N),& k+l\geq N+1 \\
              0,& k+l<N+1
                \end{array} \right.
$$
Let $A$ an order $N$ diagonal matrix. Clearly the matrices $\nabla^k A$ and  $\Delta^k A$ have level $-k$  and level $N-k$ respectively.
\end{remark}

\subsection{The family of Lax matrices $L^{N,k}$} \label{secdef:lax}

Consider the following family of Lax matrices of order $N\in \mathbb{N},$
\begin{align*}
L^{N,k}(X;P,\lambda):=P+X^\nabla+\lambda\, X^\Delta, \quad N\in \mathbb{N},\quad k\in\{ 1,2,\ldots, N-1\},
\end{align*}
where
\begin{align*}
X^\nabla:=\sum_{i=1}^{k}\nabla^{i}X^{(i)},&& X^\Delta:=\sum_{i=1}^{k}\Delta^i  X^{(i)},
\end{align*}
 with $X^{(j)},$ $j=1,\ldots,k$ the order $N$ diagonal matrices with  entries $(X^{(j)})_{i,i}:=x^{j,i},$ that stands for the shorthand notation of $x^{j,i}_{m+1/2,n},$ $j=1,\ldots, k,$  $i=1,\ldots,N,$ $m,n\in \mathbb{Z}.$ Also
with $\lambda$ we denote the spectral parameter and with $P$ the order $N$ diagonal matrix with entries $(P)_{i,i}:=p^i,$ where, as it was defined earlier, $p^i$ stands for the shorthand notation of $p^i_{m+1/2},$ $m\in \mathbb{Z}$.

Let us also denote with $Q$ the order $N$ diagonal matrix with entries $(Q)_{i,i}:=q^i\equiv q^i_{n+1/2},$ $i=1,\ldots, N,$ $n\in \mathbb{Z}$  and with
$Y^{(j)},$ $j=1,2,\ldots,k$ the order $N$ diagonal matrices with  entries $(Y^{(j)})_{i,i}:=y^{j,i}\equiv y^{j,i}_{m,n+1/2},$ $j=1,\ldots, k,\; i=1,2,\ldots,N,$ $m,n\in \mathbb{Z}.$ Note that  the dependent variables $x^{j,i},y^{j,i}$ are considered elements of a division ring $\mathbb{D}$ and the parameters $p^i,q^i,$ as well as the spectral parameter $\lambda$, are assumed central elements of $\mathbb{D}$ i.e. they commute with any element of $\mathbb{D}$.

The discrete Lax equation
\begin{align*}
L^{N,k}(X_2;P,\lambda)L^{N,k}(Y;Q,\lambda)=L^{N,k}(Y_1;Q,\lambda)L^{N,k}(X;P,\lambda)
\end{align*}
reads:
\begin{equation}\label{comp:1}
\begin{array}{ll}
&\left(P+X_2^{\nabla}\right)\left(Q+Y^{\nabla}\right)+\lambda\left(\left(P+X_2^{\nabla}\right)Y^{\Delta}+X_2^{\Delta}\left(Q+Y^{\nabla}\right) \right) +\lambda^2 X_2^{\Delta}Y^{\Delta}\\ [3mm]
=&\left(Q+Y_1^{\nabla}\right)\left(P+X^{\nabla}\right)+\lambda\left(\left(Q+Y_1^{\nabla}\right)X^{\Delta}+Y_1^{\Delta}\left(P+X^{\nabla}\right) \right) +\lambda^2 Y_1^{\Delta}X^{\Delta}.
\end{array}
\end{equation}
The requirement that (\ref{comp:1})  holds for all $\lambda$ implies
\begin{align} \label{eq0}
\left(P+X_2^{\nabla}\right)\left(Q+Y^{\nabla}\right)=\left(Q+Y_1^{\nabla}\right)\left(P+X^{\nabla}\right),\\ \label{eq1}
\left(P+X_2^{\nabla}\right)Y^{\Delta}+X_2^{\Delta}\left(Q+Y^{\nabla}\right) = \left(Q+Y_1^{\nabla}\right)X^{\Delta}+Y_1^{\Delta}\left(P+X^{\nabla}\right), \\ \label{eq2}
X_2^{\Delta}Y^{\Delta}=Y_1^{\Delta}X^{\Delta}.
\end{align}

\begin{prop}
The matrix equations (\ref{eq0})-(\ref{eq2}) consist of $2kN$ scalar equations in total.
\end{prop}

\begin{proof}
First we count the number of scalar equations of (\ref{eq0}). For $N$  even and $k\leq N/2,$ making use of Remark  (\ref{rem1}) we  deduce that the non-zero entries of the matrix equation (\ref{eq0}) are at levels $-1,-2,\ldots,-2k$ so there are $\sum_{i=1}^{2k}(N-i)$ scalar equations, while for  $k>N/2,$ (\ref{eq0}) is strictly lower triangular so it consists of $\frac{N(N-1)}{2}$ equations. Similarly, when $N$ is odd  and $k\leq \frac{N-1}{2}$  equation (\ref{eq0}) consists of  $\sum_{i=1}^{2k}(N-i)$ scalar equations, while for $k>\frac{N-1}{2},$ it consists of $\frac{N(N-1)}{2}$ equations.

When $N$ is even and $K\leq N/2,$ equation (\ref{eq1}) has non-zero entries at levels $N-1,\ldots, N-2k$ so  consists of $k(2k+1)$ scalar equations, while for $k>N/2,$ the non-zero entries exist at the upper triangular part of (\ref{eq1}) as well as to the levels $0,-1,\ldots,N-2k,$ so there are $\frac{N(N+1)}{2}+\sum_{i=1}^{2k-N}(N-i)$  scalar equations. Similarly, when $N$ is odd  and $k\leq \frac{N-1}{2}$  equation (\ref{eq1}) consists of  $k(2k+1)$ scalar equations, while for $k>\frac{N-1}{2},$ it consists of $\frac{N(N+1)}{2}+\sum_{i=1}^{2k-N}(N-i)$ equations.

Finally, equation (\ref{eq2}), when $N$ is even and $k\leq N/2$ consists of zero scalar equations since $ X^{\Delta}_2Y^{\Delta}$ and $ Y_1^{\Delta}X^{\Delta} $ are both the zero matrices, while for $k>N/2,$ it consists of $\sum_{i=1}^{2k-N}i$ equations. When $N$ is odd and $k\leq \frac{N-1}{2},$ we have zero scalar equations, but for $k>\frac{N-1}{2},$ we have  $\sum_{i=1}^{2k-N}i$ scalar equations.

We summarize the results above in Table \ref{table:1}, where it is evident that the total number of scalar equations in  (\ref{eq0})-(\ref{eq2}) is always $2kN.$
 \begin{table}[htb]
\captionof{table}{ Total number of scalar equations in (\ref{eq0})-(\ref{eq2})} \label{table:1}
\begin{tabular}{|c|c|c|c|c|}
  \hline
& \multicolumn{2}{c|}{$N$ odd  }& \multicolumn{2}{c|}{  $N$ even }  \\ \hline
  &  $  {\dis k\leq  \frac{N-1}{2}}$ &${\dis k>\frac{N-1}{2}}$ & $k\leq N/2$& $k>N/2$\\ [2mm]  \hline
  $\#$ of eqs. in (\ref{eq0}) &  ${\dis \sum_{i=1}^{2k}(N-i)}$ & ${\dis\frac{N(N-1)}{2}}$ & ${\dis \sum_{i=1}^{2k}(N-i)}$ & ${\dis\frac{N(N-1)}{2}}$ \\[1mm] \hline
  $\#$ of eqs. in (\ref{eq1}) & $k(2k+1)$ & $\begin{array}{l}
                                            \vspace{-0.1pt}{\dis \frac{N(N+1)}{2}}\\ [2mm]
                                           \;\; +{\dis \sum_{i=1}^{2k-N}(N-i)}
                                            \end{array}$ & $k(2k+1)$ & $\begin{array}{l}
                                            {\dis\frac{N(N+1)}{2}}\\ [2mm]
                                            \;\;+{\dis\sum_{i=1}^{2k-N}(N-i)}
                                            \end{array}$ \\ [1mm] \hline
  $\#$ of eqs. in (\ref{eq2}) & $0$ & ${\dis\sum_{i=1}^{2k-N}i}$ & $0$ & ${\dis\sum_{i=1}^{2k-N}i}$ \\[1mm] \hline
   total $\#$ of eqs.  & $2 k N$ & $2 k N$ & $2 k N$ & $2 k N$ \\
   \hline
\end{tabular}
\end{table}

\end{proof}
\begin{corollary}
The matrix equations (\ref{eq0})-(\ref{eq2}) constitute a difference system of $2kN$ scalar equations on the $2kN$ edge variables $x^{j,i}, y^{j,i},$ $j=1,\ldots k,$ $i=1,\ldots, N,$  $k\in \{1,\ldots,N-1\}.$
\end{corollary}
In the example that follows we present the explicit form of the Lax matrices $L^{2,1}(X;P,\lambda),$ $L^{3,1}(X;P,\lambda),$ $L^{3,2}(X;P,\lambda)$ and $L^{4,2}(X;P,\lambda)$.
\begin{example}
The Lax matrices $L^{2,1}(X;P,\lambda),$ $L^{3,1}(X;P,\lambda),$ $L^{3,2}(X;P,\lambda)$ and $L^{4,2}(X;P,\lambda),$ respectively read:
\begin{align*}
L^{2,1}(X;P,\lambda)=P+\nabla^1 X^{(1)}+\lambda\, \Delta^1 X^{(1)}=\begin{pmatrix}
                      p^1&\lambda\; x^{1,2}\\
                      x^{1,1}&p^2
                     \end{pmatrix},
\end{align*}
\begin{align*}
L^{3,1}(X;P,\lambda)=P+\nabla^1 X^{(1)}+\lambda\, \Delta^1 X^{(1)}=\begin{pmatrix}
                      p^1&0&\lambda\, x^{1,3}\\
                      x^{1,1}&p^2& 0 \\
                      0  &x^{1,2}&p^3
                     \end{pmatrix},
\end{align*}
\begin{align*}
L^{3,2}(X;P,\lambda)=P+\nabla^1 X^{(1)}+\nabla^2 X^{(2)}+\lambda \left(\Delta^1 X^{(1)}+\Delta^2 X^{(2)}\right)=\begin{pmatrix}
                      p^1&\lambda\, x^{2,2}&\lambda\, x^{1,3}\\
                      x^{1,1}&p^2& \lambda\, x^{2,3} \\
                      x^{2,1}  &x^{1,2}&p^3
                     \end{pmatrix},
\end{align*}
\begin{align*}
L^{4,2}(X;P,\lambda)=P+\nabla^1 X^{(1)}+\nabla^2 X^{(2)}+\lambda \left(\Delta^1 X^{(1)}+\Delta^2 X^{(2)}\right)=\begin{pmatrix}
                      p^1&0&\lambda\, x^{2,3}&\lambda\, x^{1,4}\\
                      x^{1,1}&p^2& 0&\lambda\, x^{2,4} \\
                      x^{2,1}  &x^{1,2}&p^3 &0\\
                      0        &x^{2,2} &x^{1,3}&p^4
                     \end{pmatrix}.
\end{align*}
\end{example}
\section{The Lax matrix $L^{N,1}$ and integrable hierarchies of difference systems} \label{Section:3}
For $k=1,$ the Lax matrix $L^{N,1}(X;P,\lambda)$ explicitly reads:
\begin{equation}\label{Lax:eq0}
L^{N,1}(X;P,\lambda)=P+\nabla^1 X^{(1)}+\lambda \Delta^1 X^{(1)}=\begin{pmatrix}
                      p^1&0&\cdots&0&\lambda\,x^{N}\\
                      x^1&p^2& 0& \cdots &0 \\
                      0  &x^2&\ddots& {}  &\vdots\\
                      \vdots & &\ddots &p^{N-1} &0\\
                       0     & 0 &   &  x^{N-1}    &p^N
                     \end{pmatrix},
\end{equation}
where we have written  $x^i$ instead of $x^{1,i}$ for the entries of the matrix $X^{(1)}$. The compatibility conditions (\ref{eq0})-(\ref{eq2}) read:
\begin{align}
\begin{aligned}\label{Lax:eq1}
x^i_2 y^{i-1}=y^i_1 x^{i-1},
\end{aligned}\\
\begin{aligned}\label{Lax:eq2}
q^i x^i_2-q^{i+1}x^i=p^iy^i_1-p^{i+1}y^i,
\end{aligned}
\end{align}
where the superscripts $i=1,2,\ldots, N,$ are considered modulo $N$.

 For arbitrary $N,$
  system (\ref{Lax:eq1}),(\ref{Lax:eq2}) under the change of the dependent variables
     \begin{align} \label{cv:1}
    (x^i_2,x^i,y^i_1,y^i)\mapsto \left(\frac{q^{i+1}p^{i+1}}{q^i}x^i_2,p^{i+1}x^{i+1},\frac{q^{i+1}p^{i+1}}{p^i}y^i_1,q^{i+1}y^{i+1}\right), & & \forall i \in\{1,\ldots, N\},
    \end{align}
    is mapped to the $N-$periodic reduction of the so-called {\em non-commutative KP map} introduced in \cite{Doliwa_2013,Doliwa_2014}. The change of variables (\ref{cv:1}) though is not admissible (see Definition \ref{adm:Mob} and Example \ref{example:1} for discussion).

\subsection{An integrable hierarchy of difference systems in edge variables}
\begin{prop}
Matrix (\ref{Lax:eq0}), serves as a strong Lax matrix for the hierarchy of difference systems in edge variables (\ref{Lax:eq1}),(\ref{Lax:eq2}) that in its solved form reads
\begin{align} \label{Lax:eq3}
\begin{aligned}
x^i_2=(q^{i+1}x^i-p^{i+1}y^i)x^{i-1}(q^ix^{i-1}-p^iy^{i-1})^{-1},\\
y^i_1=(q^{i+1}x^i-p^{i+1}y^i)y^{i-1}(q^ix^{i-1}-p^iy^{i-1})^{-1},
\end{aligned}& & i=1,2,\ldots, N.
\end{align}
\end{prop}
\begin{proof}
Substituting (\ref{Lax:eq3}) to the compatibility conditions(\ref{Lax:eq1}),(\ref{Lax:eq2}), one can easily check that the latter are satisfied.
\end{proof}

\begin{corollary} \label{lemma:3.20}
The products $\mathcal{P}^i:=x^{N+i-1}\cdots x^{k+1} x^k$ and $\mathcal{Q}^i:=y^{N+i-1}\cdots y^{k+1} y^k,$ $i=1,2,\ldots,N$ satisfy the relations:
\begin{align} \label{Lax:pr1}
\begin{aligned}
\mathcal{P}^i_2=(q^{i+1}x^i-p^{i+1}y^i)\mathcal{P}^{i+1}(q^{i+1}x^i-p^{i+1}y^i)^{-1},\\ \mathcal{Q}^i_2=(q^{i+1}x^i-p^{i+1}y^i)\mathcal{Q}^{i+1}(q^{i+1}x^i-p^{i+1}y^i)^{-1},
\end{aligned} && i=1,2,\ldots,N.
\end{align}

\end{corollary}
\begin{proof}
The direct substitution of (\ref{Lax:eq3}) into (\ref{Lax:pr1}) validates the formulae.
\end{proof}

 In the commutative setting where all variables are considered elements of the center of the algebra $\mathbb{D}$,  the first set of equations of Corollary \ref{lemma:3.20} reduce to $\prod_{i=1}^N x^i_2=\prod_{i=1}^N x^i,$ that suggests that $\prod_{i=1}^N x^i$ does not change along the set of vertical edges $E_V,$ i.e. it is function that depends only on the independent variable $m.$    We denote this function as $p^0$ to the power $N$ for reasons that will become clear later.  Similarly we can show that $\prod_{i=1}^N y^i$ is function that depends on the independent variable $n$ only that we denote as $(q^0)^N$. Namely  we have
\begin{align} \label{Lax:eq3.1}
\prod_{i=1}^N x^i=(p^0)^N=\prod_{i=1}^N x_2^i,\quad \prod_{i=1}^N y^i= (q^0)^N=\prod_{i=1}^N y^i_1.
\end{align}

 In the non-commutative setting, if we assume that the products $\mathcal{P}^i, \mathcal{Q}^i$ belong to the center of the algebra $\mathbb{D}\footnote{A less restrictive assumption  is that the products $\mathcal{P}^i, \mathcal{Q}^i$ commute with the elements $q^{i}x^{i-1}-p^iy^{i-1},$ $i=1,\ldots,N.$ },$ then from (\ref{Lax:pr1}) it follows that $\mathcal{P}^i$ are functions that depends  on the independent variable $m$ only and $\mathcal{Q}^i$ are functions that depends  on the independent variable $n$ only. We denote these functions as $(\mathrm{p}^{i-1})^N$ and $(\mathrm{q}^{i-1})^N$ respectively, namely
\be \label{Lax:eq3.21}
x^{N+i-1}\cdots x^{k+1} x^k=(\mathrm{p}^{i-1})^N, \quad y^{N+i-1}\cdots y^{k+1} y^k=(\mathrm{q}^{i-1})^N, \;\;i=1,\ldots,N.
\ee
Moreover we have that $\mathrm{p}^{i}=\mathrm{p}^{j},$   $i,j=0,1,\ldots,N-1$, that follows  if we eliminate one variable from any pair of products of the first set of equations of (\ref{Lax:eq3.21}). Similarly we  show that $\mathrm{q}^{i}=\mathrm{q}^{j},$ $i,j=0,1,\ldots,N-1.$ From further on, when we refer to the centrality assumption, we refer to the formulas:
\be \label{Lax:eq3.2}
x^{N}\cdots x^{2} x^1=(p^{0})^N=x^{N}_2\cdots x^{2}_2 x^1_2, \quad y^{N}\cdots y^{2} y^1=(q^{0})^N=y^{N}_1\cdots y^{2}_1 y^1_1, \;\;i=1,\ldots,N.
\ee
The centrality assumptions were first introduced in \cite{Doliwa_2013,Doliwa_2014} for the $N-$periodic reduction of the KP-map and they play here as well a  crucial role to the quadrirationality of the hierarchy of difference systems in edge variables (\ref{Lax:eq3}), as it is shown in the Proposition that follows.
\begin{prop}\label{prop1.2}
The hierarchy of difference systems (\ref{Lax:eq3}) is  birational. If  the centrality assumptions (\ref{Lax:eq3.2}) are imposed, then (\ref{Lax:eq3}) is  quadrirational.
\end{prop}
\begin{proof}
The polynomial form of  (\ref{Lax:eq3}), namely the equations (\ref{Lax:eq1}), (\ref{Lax:eq2}),  can be solved rationally for $x^i,y^i$ in terms of $x^i_2,y^i_1$ and that proves biratonality  of the system. Specifically from (\ref{Lax:eq1}), (\ref{Lax:eq2}) we obtain:
\begin{align} \label{Lax:eq3inv}
\begin{aligned}
x^i=(q^{i+1}x_2^{i+1}-p^{i+1}y_1^{i+1})^{-1}x_2^{i+1}(q^ix_2^{i}-p^iy_1^{i}),\\
y^i=(q^{i+1}x_2^{i+1}-p^{i+1}y_1^{i+1})^{-1}y_1^{i+1}(q^ix_2^{i}-p^iy_1^{i}),
\end{aligned}& & i=1,2,\ldots, N.
\end{align}

As it was stated in Definition \ref{def2}, if the birational system (\ref{Lax:eq3}) can be solved rationally for $x_2^i,y^i$ in terms of $x^i,y^i_1$ and the resulting system is birational, then system (\ref{Lax:eq3}) will be called {\em quadrirational} and the resulting birational system will be called {\em  companion system}.
Now we are ready to prove that (\ref{Lax:eq3}) is quadrirational provided that the centrality assumptions (\ref{Lax:eq3.2}) are imposed.

From (\ref{Lax:eq1}) we obtain
 \begin{align} \label{prop:32:1}
x_2^{i}=y_1^ix^{i-1}(y^{i-1})^{-1},
\end{align}
 substituting these equations into (\ref{Lax:eq2}) we get
 \begin{align} \label{prop:32:2}
p^{i+1}y^iy^{i-1}-(q^{i+1}x^i+p^iy^i_1)y^{i-1}+q^iy^i_1x^{i-1}=0.
\end{align}
Multiplying equation (\ref{prop:32:2}) respectively with $y^{i-2}, y^{i-2} y^{i-3},\ldots, \prod_{k=2}^{N-1}y^{i-k}$ we have
\be \label{prop:32:3}
\begin{array}{c}
{\displaystyle p^{i+1}y^iy^{i-1}y^{i-2}-(q^{i+1}x^i+p^iy^i_1)y^{i-1}y^{i-2}+q^iy^i_1x^{i-1}y^{i-2}}=0,\\[3mm]
{\displaystyle p^{i+1}y^iy^{i-1}y^{i-2} y^{i-3}-(q^{i+1}x^i+p^iy^i_1)y^{i-1}y^{i-2} y^{i-3}+q^iy^i_1x^{i-1}y^{i-2} y^{i-3}}=0,\\[3mm]
\vdots\\
{\displaystyle p^{i+1}y^iy^{i-1}\prod_{k=2}^{N-1}y^{i-k}-(q^{i+1}x^i+p^iy^i_1)y^{i-1}\prod_{k=2}^{N-1}y^{i-k}+q^iy^i_1x^{i-1}\prod_{k=2}^{N-1}y^{i-k}}=0.
\end{array}
\ee
Assuming centrality we have that $y^iy^{i-1}\prod_{k=2}^{N-1}y^{i-k}=({q}^0)^N$ and the last equation of (\ref{prop:32:3}) reads
\be \label{prop:32:4}
p^{i+1}({q}^0)^N-(q^{i+1}x^i+p^iy^i_1)y^{i-1}\prod_{k=2}^{N-1}y^{i-k}+q^iy^i_1x^{i-1}\prod_{k=2}^{N-1}y^{i-k}=0.
\ee
The products $y^{i-1}\prod_{k=2}^{N-1}y^{i-k},$  $\prod_{k=2}^{N-1}y^{i-k}$ in (\ref{prop:32:4}) can be determined recursively from (\ref{prop:32:2}), (\ref{prop:32:3}) and they are linear polynomials on $y^{i-N+1}.$ Letting $i$ run from $1$ to $N$  we obtain $y^i$ as functions of $x^i, y^i_1,$  and together with  (\ref{prop:32:1}), we have solved   (\ref{Lax:eq3}) rationally  for $x_2^i,y^i$ in terms of $x^i,y^i_1.$ In exactly similar manner we can solve (\ref{Lax:eq3}) rationally  for $x^i,y^i_1$ in terms of $x_2^i,y^i$ and that completes the proof.
\end{proof}
\subsubsection{Multidimensional compatibility}

Under the identifications $X^{i,a}:=x^{i},\; X_b^{i,a}:=x_2^{i}, \;X^{i,b}:=y^i, \;X_a^{i,b}:=y_1^i,$ the hierarchy of difference systems (\ref{Lax:eq3}) obtains the compact form
\begin{align}\label{Lax:eq4}
X^{i,a}_b=\left(p^{i+1,b}X^{i,a}-p^{i+1,a}X^{i,b}\right)X^{i-1,a}\left(p^{i,b} X^{i-1,a}-p^{i,a}X^{i-1,b}\right)^{-1},
\end{align}
$i=1,\ldots,N,$ $a\neq b\in\{1,2\}.$
\begin{lemma}\label{Lax:lemma1}
It holds
 \begin{align}  \label{Lax:lemma1:eq1}
p^{i+1,b}X^{i,a}_c-p^{i+1,a}X^{i,b}_c=K^{i,a}_{bc}\left(p^{i,c}X^{i-1,a}-p^{i,a}X^{i-1,c}\right)^{-1},
 \end{align}
$i=1,\ldots,N,\; a\neq b \neq c \neq a\in\{1,\ldots,n\},$ where
 \begin{align*}
K^{i,a}_{bc}:=\left(A\left(p^{i,c}X^{i-1,b}-p^{i,b}X^{i-1,c}\right)^{-1}\right)X^{i-1,a}+B,
 \end{align*}
with
\begin{align*}
\begin{split}
A:=&p^{i+1,b}p^{i+1,c}X^{i,a}\left(p^{i,c}X^{i-1,b}-p^{i,b}X^{i-1,c}\right)+p^{1+i,a}p^{i,b}p^{1+i,b}X^{i,c}X^{i-1,c}\\
    &-p^{1+i,a}p^{i,c}p^{1+i,c}X^{i,b}X^{i-1,b},
\end{split}    \\
B:=&p^{i,a}p^{i+1,a}\left(p^{i+1,c}X^{i,b}-p^{i+1,b}X^{i,c}\right)\left(p^{i,c}(X^{i-1,c})^{-1}-p^{i,b}(X^{i-1,b})^{-1}\right)^{-1},
\end{align*}
and the functions $K^{i,a}_{bc}$ are symmetric under the interchange $b\leftrightarrow c$ of the discrete shifts i.e.
 $K^{i,a}_{bc}=K^{i,a}_{cb}$.
\end{lemma}
\begin{proof}
Substituting from (\ref{Lax:eq4}) the expressions of $X^{i,a}_c$ and $X^{i,b}_c$ into $p^{i+1,b}X^{i,a}_c-p^{i+1,a}X^{i,b}_c,$ upon expansion, recollection of terms  and making use of the identity
 \begin{align*}
X^{i-1,b}\left(p^{i,c}X^{i-1,b}-p^{i,b}X^{i-1,c}\right)^{-1}X^{i-1,c}=\left(p^{i,c}(X^{i-1,c})^{-1}-p^{i,b}(X^{i-1,b})^{-1}\right)^{-1},
 \end{align*}
we obtain (\ref{Lax:lemma1:eq1}). The expressions of $B$ and of $A\left(p^{i,c}X^{i-1,b}-p^{i,b}X^{i-1,c}\right)^{-1}$ are clearly symmetric under the interchange $b\leftrightarrow c,$  hence there is $K^{i,a}_{bc}=K^{i,a}_{cb}$.
\end{proof}
\begin{prop} \label{prop:3.4}
The system of difference equations (\ref{Lax:eq4}) can be extended in a compatible way to $n-$dimensions as follows
 \begin{align} \label{Lax:eq5}
X^{i,a}_b=\left(p^{i+1,b}X^{i,a}-p^{i+1,a}X^{i,b}\right)X^{i-1,a}\left(p^{i,b} X^{i-1,a}-p^{i,a}X^{i-1,b}\right)^{-1},
\end{align}
with $i=1,\ldots,N,\; a\neq b\in\{1,\ldots,n\}$. 
   The compatibility conditions
 \begin{align*}
 X^{i,a}_{bc}=X^{i,a}_{cb},\quad i=1,\ldots,N,\;\; a\neq b \neq c \neq a\in\{1,\ldots,n\}
 \end{align*}
  hold.
\end{prop}
\begin{proof}
Shifting (\ref{Lax:eq5}) at the $c-$direction, we obtain
 \begin{align*}
X^{i,a}_{bc}=\left(p^{i+1,b}X^{i,a}_c-p^{i+1,a}X^{i,b}_c\right)X^{i-1,a}_c\left(p^{i,b} X^{i-1,a}_c-p^{i,a}X^{i-1,b}_c\right)^{-1}.
 \end{align*}
Substituting from (\ref{Lax:eq5}) the expression of $X^{i-1,a}_c,$ $X^{i-1,b}_c,$ $X^{i,a}_c$ and $X^{i,b}_c$   to the equations above
and by making use of Lemma (\ref{Lax:lemma1}) we derive the following
 compatibility formula
 \begin{align*}
 X^{i,a}_{bc}= K^{i,a}_{bc}X^{i-1,a}\left(K^{i-1,a}_{bc}\right)^{-1},
 \end{align*}
that is clearly symmetric under the interchange  $b$ to $c$ and that completes the proof.
\end{proof}
\begin{example}[$N=2$] \label{example:1}
For $N=2,$ we have the Lax matrix $L^{2,1}(x^1,x^2;p^1,p^2,\lambda)$ that reads:
\begin{align*}
L^{2,1}(x^1,x^2;p^1,p^2,\lambda)=\begin{pmatrix}
p^1&\lambda\, x^2\\
x^1&p^2
\end{pmatrix}
\end{align*}
and it serves as a strong Lax matrix for the birational difference system (\ref{Lax:eq3}) with $N=2$.
The centrality assumption (\ref{Lax:eq3.2}) that now reads
\begin{align*}
x^2 x^1=({p^0})^2=x^2_2 x^1_2,\quad y^2 y^1=({q^0})^2=y^2_1 y^1_1,
\end{align*}
allow us to eliminate $x^2,y^2$ and their respective shifts from (\ref{Lax:eq3}) with $N=2$. For simplicity we denote $x:=x^1,$ $y:=y^1,$ $p:=(p^0)^2$ and $q:=(q^0)^2$   to obtain the quadrirational difference system
\begin{align} \label{chIIIa}
x_2=p(q^2x-p^2y)(pq^1y-qp^1x)^{-1}y,\quad y_1=q(q^2x-p^2y)(pq^1y-qp^1x)^{-1}x,
\end{align}
with strong Lax matrix
\begin{align*}
L^{2,1}(x^1;p,p^1,p^2,\lambda)=\begin{pmatrix}
p^1&\lambda\, p\,(x^1)^{-1}\\
x^1&p^2
\end{pmatrix}.
\end{align*}
Solving (\ref{chIIIa}) for $x,y_1$ in terms of $x_2,y$ and by applying the following admissible transformation
\begin{align*}
(x,y,x_2,y_1)\mapsto (x^{-1},y^{-1},x_2^{-1},y_1^{-1}),
\end{align*}
for $p^1=p^2=q^1=q^2=1$ we obtain:
\begin{align} \label{hIIIa}
x=(p)^{-1}y(y+x_2)^{-1}(px_2+qy),\quad y_1=(q)^{-1}x_2(y+x_2)^{-1}(px_2+qy).
\end{align}
The difference system (\ref{hIIIa}) defines a parametric Yang-Baxter map $R^{(p,q)}:\mathbb{D}\times \mathbb{D}\ni(x_2,y)\mapsto (x,y_1)\in \mathbb{D}\times \mathbb{D},$ since it holds
\begin{align} \label{YB:eq}
R^{(p,q)}_{12}\circ R^{(p,r)}_{13}\circ R^{(q,r)}_{23}=R^{(q,r)}_{23}\circ R^{(p,r)}_{13}\circ R^{(p,q)}_{12},
\end{align}
where $R^{(p,q)}_{12}:\mathbb{D}\times \mathbb{D}\times \mathbb{D}\mapsto \mathbb{D}\times \mathbb{D}\times \mathbb{D}$ is defined as the map that as $R^{(p,q)}$ on the first and second component of $\mathbb{D}\times \mathbb{D}\times \mathbb{D}$ and as an identity to the third; Similarly are defined the maps $ R^{(p,r)}_{13}$ and  $R^{(q,r)}_{23}$. The first instances of {\em set-theoretical-solutions} of the quantum Yang-Baxter equation (\ref{YB:eq}) appeared in \cite{Sklyanin:1988,Drinfeld:1992}  and the term {\em Yang-Baxter maps} was coined to these set-theoretical-solutions in \cite{Bukhshtaber:1998,Veselov:20031}.

 The Yang-Baxter map $R^{(p,q)}$ provided by (\ref{hIIIa}),   serves as the non-commutative extension of the $H_{III}^A$ Yang-Baxter map that was introduced in \cite{Papageorgiou:2010} in the commutative setting.

Note that (\ref{Lax:eq1}),(\ref{Lax:eq2}) under the non-admissible change of variables (\ref{cv:1}) coincides with the $N-$periodic reduction of the KP map \cite{Doliwa_2013,Doliwa_2014}, that for $N=2$ and under the centrality assumptions, results to non-commutative extension \cite{Doliwa_2014} of the $H_{III}^B$ Yang-Baxter map that was introduced in \cite{Papageorgiou:2010} in the commutative setting. The Yang-Baxter maps $H_{III}^B$ and $H_{III}^A$ are non-equivalent up to $(M\ddot ob)^2$ transformations \cite{Papageorgiou:2010}, and that justifies why we consider that the $N-$periodic reduction of the KP map is not equivalent to the difference system (\ref{Lax:eq1}),(\ref{Lax:eq2}).
\end{example}

\begin{example}[$N=3$] \label{example:1.2}
For $N=3,$ we have the following strong Lax matrix
\begin{align*}
L^{3,1}(x^1,x^2,x^3;p^1,p^2,p^3,\lambda)=\begin{pmatrix}
p^1&0&\lambda\, x^3\\
x^1&p^2&0\\
0&x^2&p^3
\end{pmatrix},
\end{align*}
that serves as a Lax matrix for the birational difference system (\ref{Lax:eq3}) with $N=3$.
  The centrality assumption (\ref{Lax:eq3.2}) that now reads
\begin{align*}
x^3 x^2 x^1=(\mathrm{p^0})^3=x^3_2x^2_2 x^1_2,\quad y^3y^2 y^1=(\mathrm{q^0})^3=y^3_1y^2_1 y^1_1,
\end{align*}
allow us to eliminate $x^3,y^3$ and their respective shifts from (\ref{Lax:eq3}) with $N=3$ and obtain the quadrirational difference system
\begin{align} \label{CHIIIa}
\begin{aligned}
x^i_2=(q^{i+1}x^i-p^{i+1}y^i)x^{i-1}(q^ix^{i-1}-p^iy^{i-1})^{-1},\\
y^i_1=(q^{i+1}x^i-p^{i+1}y^i)y^{i-1}(q^ix^{i-1}-p^iy^{i-1})^{-1},
\end{aligned}
\end{align}
where  the index $i$ is considered modulo $3$ and take the values   $i=1,2,$ with $x^3\equiv (p^0)^3(x^1)^{-1}(x^2)^{-1}$ and $y^3\equiv (q^0)^3(y^1)^{-1}(y^2)^{-1}.$ The quadrirational system (\ref{CHIIIa})
has as strong Lax matrix the matrix
\begin{align*}
L^{3,1}(x^1,x^2;p^0,p^1,p^2,p^3,\lambda)=\begin{pmatrix}
p^1&0&\lambda\, (p^0)^3 (x^1)^{-1} (x^2)^{-1}\\
x^1&p^2&0\\
0&x^2&p^3
\end{pmatrix}.
\end{align*}
Following the method introduced in the proof of Proposition \ref{prop1.2}, we solve (\ref{CHIIIa}) for $x^1_2,x^2_2,y^1,y^2$ in terms of $x^1,x^2,y^1_1,y^2_1$  to obtain the companion system:
\begin{align} \label{HIIIa}
\begin{aligned}
x_2^i=\left(p^iy_1^ix^{i+2}S^{i+1}+q^{i+1}q^{i+2}(p^0)^3\right) \left(q^ix^{i+2}S^{i+1}+p^{i+1}p^{i+2}y^{i+2}_1y^{i+1}_1\right)^{-1} \\
y^i= \left(p^{i+1}S^{i+2}y^{i+1}_1+q^{i+1}q^{i+2}x^{i+2}x^{i+1}\right)^{-1}\left(q^{i+1}S^{i+2}y^{i+1}_1x^i+p^ip^{i+2}(q^0)^3\right),
\end{aligned}
\end{align}
where  the index $i$ is considered modulo $3$ and take the values   $i=1,2,$ with $x^3\equiv (p^0)^3(x^1)^{-1}(x^2)^{-1},$  $y^3\equiv (q^0)^3(y^1)^{-1}(y^2)^{-1}$ and the expressions $S^j$ are defined by $S^j:=q^{j+1}x^j+p^jy^j_1.$

The difference system (\ref{HIIIa}) defines a Yang-Baxter map
\begin{align*}
R:\mathbb{D}\times \mathbb{D}\times \mathbb{D}\times \mathbb{D}\ni(x^1,x^2,y^1_1,y^2_1)\mapsto (x^1_2,x^2_2,y^1,y^2)\in \mathbb{D}\times \mathbb{D}\times \mathbb{D}\times \mathbb{D},
\end{align*}
that it in the simple case where $p^1=p^2=p^3=1=q^1=q^2=q^3,$ can be considered as the  second member of the hierarchy of  $H_{III}^A$ Yang-Baxter maps extended in the non-commutative domain..
\end{example}

\subsection{Integrable hierarchies of difference systems in vertex variables}

A procedure to obtain integrable difference systems in vertex variables from integrable systems in edge (or even face variables) was introduced in \cite{doliwa-santini}. Nowadays the name {\em potentialisation} is coined to this procedure and it is widely applied \cite{Kassotakis_2011,Kassotakis_2012,Doliwa_2013,Kassotakis_2018,Fordy_2017,Kassotakis_2021}.
In order to obtain integrable difference systems in vertex variables associated with the integrable difference system  (\ref{Lax:eq3}), we  follow the {\em potentialisation} procedure as it was used in \cite{Doliwa_2013,Fordy_2017}.

\subsubsection{Multiplicative potentials and the lattice-modified Gel'fand-Dikii hierarchy}
The hierarchy of difference systems (\ref{Lax:eq3}) in its polynomial form constitutes of the sets of equations (\ref{Lax:eq1}), (\ref{Lax:eq2}). 
Equations (\ref{Lax:eq1}) are identically satisfied if we set
\be \label{Lax:eq6}
x^i=p^0 \phi^i_1(\phi^{i-1})^{-1}\quad y^i=q^0 \phi^i_2(\phi^{i-1})^{-1}, \quad i=1,\ldots,N.
\ee
The functions $\phi^i,$ $i=1,\ldots,N$ can be considered as potential functions. In terms of these potential functions equations (\ref{Lax:eq1}) are identically satisfied, while equations (\ref{Lax:eq2}) read
\begin{align} \label{eq:pot1}
q^0\left(p^{i+1}\phi^{i}_2\left(\phi^{i-1}\right)^{-1}-p^i\phi^{i}_{12}\left(\phi^{i-1}_1\right)^{-1}\right)-p^0 \left(q^{i+1}\phi^{i}_1\left(\phi^{i-1}\right)^{-1}-q^i\phi^{i}_{12}\left(\phi^{i-1}_2\right)^{-1}\right)=0,
\end{align}
$i=1,\ldots, N$ and constitute a hierarchy of difference systems in vertex variables.
\begin{prop} \label{prop:mpot}
The hierarchy of difference systems in vertex variables (\ref{eq:pot1}), after rearranging terms reads
\begin{align} \label{eq:pot1:sol}
\left(q^0p^{i+1}\phi^i_2-p^0q^{i+1}\phi^i_1 \right)\left(\phi^{i-1} \right)^{-1}=\phi^i_{12}\left(q^0p^i\left(\phi^{i-1}_1 \right)^{-1}-p^0q^i\left(\phi_2^{i-1} \right)^{-1} \right),\;\; i=1,\ldots, N,
\end{align}
\begin{enumerate}
\item arises as the compatibility conditions of the Lax equation (\ref{lax-eq-ver}), associated with the strong Lax matrix
\begin{align*}
L(\Phi_1,\Phi;P,p^0,\lambda)= \begin{pmatrix}
                      p^1&0&\cdots&0&\lambda\, p^0\phi^N_1\left(\phi^{N-1}\right)^{-1}\\
                      p^0\phi^1_1\left(\phi^N\right)^{-1}&p^2& 0& \cdots &0 \\
                      0  &p^0\phi^2_1\left(\phi^1\right)^{-1}&\ddots& {}  &\vdots\\
                      \vdots & &\ddots &p^{N-1} &0\\
                       0     & 0 &   &  p^0\phi^{N-1}_1\left(\phi^{N-2}\right)^{-1}    &p^N
                     \end{pmatrix};
\end{align*}
\item it is multidimensional consistent;
\item it is invariant under the following permutations of the dependent variables
\begin{align*}
\tau :(\phi^j,\phi^j_1,\phi^j_2,\phi^j_{12},p^0,q^0,p^j,q^j)\mapsto (\phi^j,\phi^j_2,\phi^j_1,\phi^j_{12},q^0,p^0,q^j,p^j),\\
\sigma : (\phi^j,\phi^j_1,\phi^j_2,\phi^j_{12},p^0,q^0,p^j,q^j)\mapsto (\phi^j_{12},\phi^j_2,\phi^j_1,\phi^j,q^0,p^0,q^j,p^j),\\
\forall j\in \{1,\ldots,N\},
\end{align*}
i.e. it respects the {\it rombic} symmetry;

\item it is an integrable hierarchy of difference systems in vertex variables.
\end{enumerate}
\end{prop}
\begin{proof}
\begin{enumerate}
\item Substituting the expressions of the potential functions (\ref{Lax:eq6}) into (\ref{Lax:eq0}) we obtain the Lax matrix presented in this Proposition.
\item This is a direct consequence of the multidimensional compatibility (see Proposition \ref{prop:3.4}) of the underlying difference system in edge variables (\ref{Lax:eq3}). Specifically, for the system  (\ref{Lax:eq5}) we have proved that the compatibility conditions
    \begin{align}\label{comp}
 X^{i,a}_{bc}=X^{i,a}_{cb},\quad i=1,\ldots,N,\;\; a\neq b \neq c \neq a\in\{1,\ldots,n\},
 \end{align}
  hold. For the potential functions $\phi^i$ we have
        \begin{align} \label{pot11}
    X^{i,a}=p^{0,a}\phi^i_a(\phi^{i-1})^{-1}, \quad i=1,\ldots,N, \quad a\in \{1,\ldots, n\}.
    \end{align}
From (\ref{pot11}) and (\ref{comp}) we obtain
\begin{align*}
\phi^i_{abc}=\phi^i_{acb}, \quad i=1,\ldots,N,\;\; a\neq b \neq c \neq a\in\{1,\ldots,n\},
\end{align*}
that proves the multidimensional consistency of the hierarchy.
\item It is easy to see that (\ref{eq:pot1:sol}) it is invariant under $\tau,$ while acting with $\sigma$ on (\ref{eq:pot1:sol}) we obtain
\begin{align} \label{eq:pot1:sold}
\left(p^0q^{i+1}\phi^i_2-q^0p^{i+1}\phi^i_1 \right)\left(\phi^{i-1}_{12} \right)^{-1}=\phi^i\left(p^0q^i\left(\phi^{i-1}_1 \right)^{-1}-q^0p^i\left(\phi_2^{i-1} \right)^{-1} \right),\;\; i=1,\ldots, N,
\end{align}
which is (\ref{eq:pot1:sol}) in disguise. Indeed, by acting on (\ref{eq:pot1:sol}) with  ${\bf T}_{-1} {\bf T}_{-2}$ we obtain (\ref{eq:pot1:sold}). Here with ${\bf T}_r$ we denote the forward shift operator on the $r-$th direction and with ${\bf T}_{-s}$ we denote the backward shift operator on the $s-$th direction i.e.
\begin{align*}
{\bf T}_1: \phi^i_{m,n}\mapsto \phi^i_{m+1,n}\equiv \phi^i_1,\;\; {\bf T}_{-1}: \phi^i_{m,n}\mapsto \phi^i_{m-1,n}\equiv \phi^i_{-1},\;\; {\bf T}_{-2}: \phi^i_{m,n}\mapsto \phi^i_{m,n-1}\equiv \phi^i_{-2},\;\; \mbox{etc}.
\end{align*}
\item Due to statements $(1)-(3)$, (\ref{eq:pot1:sol})  constitutes an integrable hierarchy of difference systems in vertex variables defined on the black-white lattice \cite{Papageorgiou:2009II,ABS:2009,Boll:2011}.
\end{enumerate}
\end{proof}
 In the commutative setting where all variables are considered as elements of the center of the algebra $\mathbb{D},$ relations (\ref{Lax:eq3.1}) in terms of the multiplicative potential functions (\ref{Lax:eq6}), lead to
    \begin{align*}
    \prod_{i=1}^N \phi^i_1=\prod_{i=1}^N \phi^i, \quad \prod_{i=1}^N \phi^i_2=\prod_{i=1}^N \phi^i,
    \end{align*}
    so the potential functions are not independent but they satisfy
    \begin{align} \label{Vertex:eq1}
    \prod_{i=1}^N \phi^i=c,
    \end{align}
    where $c$ is a constant that can be scaled to $1$. From (\ref{eq:pot1:sol}) together with (\ref{Vertex:eq1}) one can solve rationally  for any of the corner variable sets $\{\phi^1,\ldots, \phi^N\},$ $\{\phi^1_1,\ldots, \phi^N_1\},$ $\{\phi^1_2,\ldots, \phi^N_2\},$ $\{\phi^1_{12},\ldots, \phi^N_{12}\},$ that is a reminiscence of the associated quadrirational difference system in edge variables (\ref{Lax:eq3}),(\ref{Lax:eq3.2}).

 In the non-commutative setting,    (\ref{eq:pot1:sol}) can be solved only for the corner variable sets $\{\phi^1,\ldots, \phi^N\}$ and  $\{\phi^1_{12},\ldots, \phi^N_{12}\}.$ If we impose the centrality assumptions   (\ref{Lax:eq3.2}), that in terms of the potential functions (\ref{Lax:eq6}) take the form
    \begin{align} \label{Vertex:eq2}
    \begin{aligned}
    \phi^{N}_1\left(\phi^{N-1}\right)^{-1}\cdots \phi^{2}_1\left(\phi^{1}\right)^{-1} \phi^{1}_1\left(\phi^{N}\right)^{-1}=1=   \phi^{N}_{12}\left(\phi^{N-1}_2\right)^{-1}\cdots \phi^{2}_{12}\left(\phi^{1}_2\right)^{-1} \phi^{12}_2\left(\phi^{N}_2\right)^{-1},\\
    \phi^{N}_2\left(\phi^{N-1}\right)^{-1}\cdots \phi^{2}_2\left(\phi^{1}\right)^{-1} \phi^{1}_2\left(\phi^{N}\right)^{-1}=1=   \phi^{N}_{12}\left(\phi^{N-1}_1\right)^{-1}\cdots \phi^{2}_{12}\left(\phi^{1}_1\right)^{-1} \phi^{12}_2\left(\phi^{N}_1\right)^{-1},
    \end{aligned}
    \end{align}
    then (\ref{eq:pot1:sol}) can be solved  for any corner variable set as in the commutative case.

     For arbitrary $N$  in the non-commutative setting, a point equivalent form of (\ref{eq:pot1:sol}) with $p^i=~q^i=~1,$ $\forall i\in \{1,\ldots,N\}$ was introduced in \cite{Doliwa_2013,Doliwa_2014} and it was shown that it serves as the non-commutative lattice-modified Gel'fand-Dikii hierarchy. In the commutative setting, the same system appeared first in \cite{Atkinson:2012} and afterwards in \cite{Fordy_2017}.

      Note that the point transformation $\phi^i_{m+k,n+l}\mapsto (p^{i+1})^k(q^{i+1})^l\phi^i_{m+k,n+l},\;k,l\in \mathbb{Z}$ that scales off $p^i, q^i,$ $i=1,2,\ldots,N$  in (\ref{eq:pot1:sol}), re-written in terms of $x^i,y^i$ by using (\ref{Lax:eq6}), results to a non-admissible Mobi\"us transformation. So, hierarchy (\ref{eq:pot1:sol}) together with (\ref{Vertex:eq2}) can be considered as a  $2N-$parameter extension  of the non-commutative lattice-modified Gel'fand-Dikii hierarchy.

\begin{example}[$N=2$]
For $N=2,$ we have the Lax matrix $L(\Phi_1,\Phi;P,p^0,\lambda)$ that reads:
\begin{align*}
L(\Phi_1,\Phi;P,p^0,\lambda)=\begin{pmatrix}
p^1&\lambda\, p^0 \phi^2_1 \left(\phi^1\right)^{-1}\\
p^0\phi^1_1\left(\phi^2\right)^{-1}&p^2
\end{pmatrix}
\end{align*}
and it serves as a strong Lax matrix for the integrable difference system in vertex variables (\ref{eq:pot1}) with $N=2$, namely
\begin{align} \label{MVertex:21}
q^0\left(p^{2}\phi^{1}_2\left(\phi^{2}\right)^{-1}-p^1\phi^{1}_{12}\left(\phi^{2}_1\right)^{-1}\right)-p^0 \left(q^{2}\phi^{1}_1\left(\phi^{2}\right)^{-1}-q^1\phi^{1}_{12}\left(\phi^{2}_2\right)^{-1}\right)=0, \\ \label{MVertex:22}
q^0\left(p^{1}\phi^{2}_2\left(\phi^{1}\right)^{-1}-p^2\phi^{2}_{12}\left(\phi^{1}_1\right)^{-1}\right)-p^0 \left(q^{1}\phi^{2}_1\left(\phi^{1}\right)^{-1}-q^2\phi^{2}_{12}\left(\phi^{1}_2\right)^{-1}\right)=0.
\end{align}
The centrality assumption (\ref{Vertex:eq2}) that now reads
\begin{align*}
\phi^2_1\left(\phi^1\right)^{-1}\phi^1_1\left(\phi^2\right)^{-1}=1=\phi^2_{12}\left(\phi^1_2\right)^{-1}\phi^1_{12}\left(\phi^2_2\right)^{-1},\quad
\phi^2_2\left(\phi^1\right)^{-1}\phi^1_2\left(\phi^2\right)^{-1}=1=\phi^2_{12}\left(\phi^1_1\right)^{-1}\phi^1_{12}\left(\phi^2_1\right)^{-1},
\end{align*}
allow us to re-write the difference system in terms of one potential function only. Indeed, from the centrality assumption we have
\begin{align*}
\phi^2_2\left(\phi^1\right)^{-1}&=\phi^2\left(\phi^1_2\right)^{-1}, & \phi^{2}_{12}\left(\phi^{1}_1\right)^{-1}&=\phi^2\left(\phi^{1}_1\right)^{-1}\phi^1\left(\phi^{1}_{12}\right)^{-1},\\
\phi^2_1\left(\phi^1\right)^{-1}&=\phi^2\left(\phi^1_1\right)^{-1}, & \phi^{2}_{12}\left(\phi^{1}_2\right)^{-1}&=\phi^2\left(\phi^{1}_2\right)^{-1}\phi^1\left(\phi^{1}_{12}\right)^{-1},\\
\end{align*}
that allow us to eliminate $\phi^2,\phi^2_1, \phi^2_{12},$ from (\ref{MVertex:22}) to obtain
\begin{align} \label{Vertex:mKdV}
q^0\left(p^1\phi^{-1}_2\phi_{12}-p^2\phi^{-1}_1\phi\right)-p^0\left(q^1\phi^{-1}_1\phi_{12}-q^2\phi^{-1}_2\phi\right)=0,
\end{align}
where for simplicity we have denoted $\phi:=\phi^1,$ $\phi_1:=\phi^1_1,$ etc. When $p^1=p^2=1=q^1=q^2$ and in the commutative setting, equation (\ref{Vertex:mKdV}) coincides with the lattice potential modified KdV equation \cite{Bianchi:1894,Hir-sG,nij-qui-cap} or up to a point transformation to $H3^0$ in \cite{ABS}. An one-parameter extension of the  lattice potential modified KdV equation (again in the commutative setting) was introduced in \cite{nij-qui-cap}, whereas a two-parameter extension of the same equation was introduced in \cite{Kouloukas:2012}.

The non-commutative version of the lattice potential modified KdV equation was introduced in \cite{bs:2002N} and it includes two parameters. The extended non-commutative lattice potential modified KdV equation that was introduced in this Section includes six parameters. So (\ref{Vertex:mKdV}) can be considered as a non-commutative version of the $4-$parameter extension of the lattice potential modified KdV equation.
\end{example}
\begin{example}[$N\geq 3$]
For $N=3,$ we have the Lax matrix $L(\Phi_1,\Phi;P,p^0,\lambda)$ that reads:
\begin{align*}
L(\Phi_1,\Phi;P,p^0,\lambda)=\begin{pmatrix}
p^1&0&\lambda\, p^0\phi^3_1\left(\phi^2\right)^{-1}\\
p^0\phi^1_1\left(\phi^3\right)^{-1}&p^2&0\\
0&p^0\phi^2_1\left(\phi^1\right)^{-1}&p^3
\end{pmatrix}
\end{align*}
and it serves as a strong Lax matrix for the integrable difference system in vertex variables (\ref{eq:pot1}) with $N=3$. Explicitly it reads
\begin{align} \label{MVertex:31}
q^0\left(p^{2}\phi^{1}_2\left(\phi^{3}\right)^{-1}-p^1\phi^{1}_{12}\left(\phi^{3}_1\right)^{-1}\right)-p^0 \left(q^{2}\phi^{1}_1\left(\phi^{3}\right)^{-1}-q^1\phi^{1}_{12}\left(\phi^{3}_2\right)^{-1}\right)=0, \\ \label{MVertex:32}
q^0\left(p^{3}\phi^{2}_2\left(\phi^{1}\right)^{-1}-p^2\phi^{2}_{12}\left(\phi^{1}_1\right)^{-1}\right)-p^0 \left(q^{3}\phi^{2}_1\left(\phi^{1}\right)^{-1}-q^2\phi^{2}_{12}\left(\phi^{1}_2\right)^{-1}\right)=0,\\ \label{MVertex:33}
q^0\left(p^{1}\phi^{3}_2\left(\phi^{2}\right)^{-1}-p^3\phi^{3}_{12}\left(\phi^{2}_1\right)^{-1}\right)-p^0 \left(q^{1}\phi^{3}_1\left(\phi^{2}\right)^{-1}-q^3\phi^{3}_{12}\left(\phi^{2}_2\right)^{-1}\right)=0,
\end{align}
From the centrality assumption (\ref{Vertex:eq2}) with $N=3$ we have
\begin{align*}
\phi^3_2\left(\phi^2\right)^{-1}&=\phi^3\left(\phi^1_2\right)^{-1}\phi^1\left(\phi^2_2\right)^{-1}, & \phi^{3}_{12}\left(\phi^{2}_1\right)^{-1}&=\phi^3\left(\phi^{1}_1\right)^{-1}\phi^1\left(\phi^2_1\right)^{-1}\phi^2\left(\phi^{1}_{12}\right)^{-1}\phi^1_1
\left(\phi^{2}_{12}\right)^{-1},\\
\phi^3_1\left(\phi^2\right)^{-1}&=\phi^3\left(\phi^1_1\right)^{-1}\phi^1\left(\phi^2_1\right)^{-1}, & \phi^{3}_{12}\left(\phi^{2}_2\right)^{-1}&=\phi^3\left(\phi^{1}_2\right)^{-1}\phi^1\left(\phi^2_2\right)^{-1}\phi^2\left(\phi^{1}_{12}\right)^{-1}\phi^1_2
\left(\phi^{2}_{12}\right)^{-1},
\end{align*}
that allow us to eliminate $\phi^3_1,\phi^3_2, \phi^3_{12},$ from (\ref{MVertex:33}) that together with (\ref{MVertex:32}) we  obtain
\begin{align} \label{Vertex:m-Bussinesq1}
\begin{aligned}
q^0\left(p^{3}\phi^{2}_2\left(\phi^{1}\right)^{-1}-p^2\phi^{2}_{12}\left(\phi^{1}_1\right)^{-1}\right)-p^0 \left(q^{3}\phi^{2}_1\left(\phi^{1}\right)^{-1}-q^2\phi^{2}_{12}\left(\phi^{1}_2\right)^{-1}\right)=0,\\
q^0\left(p^{1}\left(\phi^{1}_2\right)^{-1}\phi^{1}\left(\phi^{2}_2\right)^{-1}-p^3\left(\phi^{1}_1\right)^{-1}\phi^1\left(\phi^{2}_1\right)^{-1}\phi^2
\left(\phi^{1}_{12}\right)^{-1}\phi^1_1\left(\phi^{2}_{12}\right)^{-1}    \right)\\
\qquad \quad -p^0\left(q^{1}\left(\phi^{1}_1\right)^{-1}\phi^{1}\left(\phi^{2}_1\right)^{-1}-q^3\left(\phi^{1}_2\right)^{-1}\phi^1\left(\phi^{2}_2\right)^{-1}\phi^2
\left(\phi^{1}_{12}\right)^{-1}\phi^1_2\left(\phi^{2}_{12}\right)^{-1}    \right)=0
\end{aligned}
\end{align}
 When $p^1=p^2=p^3=1=q^1=q^2=q^3$ and in the commutative setting, equation (\ref{Vertex:m-Bussinesq1}) serves as the lattice-modified Boussinesq equation \cite{Nijhoff:1992,JamesPhd,Atkinson:2012}, that in the non-commutative setting it was introduced in \cite{Doliwa_2013}.
So (\ref{Vertex:m-Bussinesq1}) serves as a non-commutative version of the $6-$parameter extension of the lattice-modified Boussinesq equation.

In order to recover exactly the form of the  lattice-modified Boussinesq equation introduced in \cite{JamesPhd,Atkinson:2012} but in the non-commutative setting, instead of the centrality assumption (\ref{Vertex:eq2}) with $N=3$, one should consider (mimicking the commutative setting) the restricted assumption $\phi^3\phi^2\phi^1=1$ and perform the re-potentialization
\begin{align} \label{repot}
(\phi^3,\phi^2,\phi^1)\mapsto \left((\phi^2)^{-1},\phi^2 (\phi^1)^{-1},\phi^1\right).
\end{align}
Then, (\ref{MVertex:31})-(\ref{MVertex:33}) together with (\ref{repot}) lead exactly to the desired form.

For arbitrary $N,$  the hierarchy of difference systems in vertex variables (\ref{eq:pot1}) is a non-commutative integrable  hierarchy that when the centrality assumptions (\ref{Vertex:eq2}) are imposed, it serves as   the non-commutative version of a $2N-$parameter extension of the lattice-modified Gel'fand-Dikii hierarchy. In the commutative setting the lattice-modified Gel'fand-Dikii hierarchy was introduced implicitly in \cite{Nijhoff:1992} (the first two members of the hierarchy were presented explicitly), the whole hierarchy was explicitly presented in \cite{JamesPhd,Atkinson:2012}. Moreover it includes two parameters. Note that the non-commutative version of the lattice-modified Gel'fand-Dikii hierarchy that  was introduced in \cite{Doliwa_2013}  includes two parameters as well. The extended non-commutative lattice-modified Gel'fand-Dikii hierarchy that was introduced in this Section includes $2N+2$ parameters. Furthermore, if instead of the centrality assumptions (\ref{Vertex:eq2}) we consider $\phi^N\cdots \phi^2\phi^1=1$ and perform the re-potentialization
\begin{align} \label{repot2}
(\phi^N,\phi^{N-1},\ldots,\phi^2,\phi^1)\mapsto \left((\phi^{N-1})^{-1},\phi^{N-1} (\phi^{N-2})^{-1},\ldots,\phi^2(\phi^1)^{-1},\phi^1\right),
\end{align}
then form (\ref{eq:pot1}) together with (\ref{repot2}), we  obtain a form of the non-commutative lattice-modified Gel'fand-Dikii hierarchy that  for $p^i=q^i=1,\;\forall i\in \{1,\ldots,N\},$ coincides with the form introduced in \cite{Doliwa_2013}.
\end{example}

\subsubsection{Additive potentials and the lattice-Schwarzian Gel'fand-Dikii hierarchy}

The hierarchy of difference systems (\ref{Lax:eq3}) in its polynomial form constitutes of the sets of equations (\ref{Lax:eq1}), (\ref{Lax:eq2}).
Equations (\ref{Lax:eq2}) are identically satisfied if we set
\be \label{Lax:eq7}
x^i=p^i\chi^i_1-p^{i+1}\chi^i \quad y^i=q^i\chi^i_2-q^{i+1}\chi^i, \quad i=1,\ldots,N.
\ee
The functions $\chi^i,$ $i=1,\ldots,N$ can be considered as potential functions. In terms of these potential functions, equations (\ref{Lax:eq2}) are identically satisfied, while equations (\ref{Lax:eq1}) read
\be \label{Vertex:pot2}
\left(p^i\chi^{i}_{12}-p^{i+1}\chi^{i}_2 \right)\left(q^{i-1} \chi^{i-1}_2-q^i \chi^{i-1} \right)-\left(q^i\chi^{i}_{12}-q^{i+1}\chi^{i}_1 \right)\left(p^{i-1} \chi^{i-1}_1-p^i \chi^{i-1} \right)=0,
\ee
$i=1,\ldots,N$ and constitute a hierarchy of difference systems in vertex variables.
\begin{prop} \label{prop:apot}
For the hierarchy of difference systems in vertex variables (\ref{Vertex:pot2}) it holds:
\begin{enumerate}
\item arises as the compatibility conditions of the Lax equation (\ref{lax-eq-ver}), associated with the strong Lax matrix
\begin{align*}
L(\Phi_1,\Phi;P,\lambda)= \begin{pmatrix}
                      p^1&0&\cdots&0&\lambda \left(p^N\chi^N_1-p^1\chi^N\right)\\
                      p^1\chi^1_1-p^2\chi^1&p^2& 0& \cdots &0 \\
                      0  &p^2\chi^2_1-p^3\chi^2&\ddots& {}  &\vdots\\
                      \vdots & &\ddots &p^{N-1} &0\\
                       0     & 0 &   &  p^{N-1}\chi^{N-1}_1-p^N\chi^{N-1}    &p^N
                     \end{pmatrix};
\end{align*}
\item it is multidimensional consistent;
\item it is invariant under the following permutations of the dependent variables
\begin{align*}
\tau :(\phi^j,\phi^j_1,\phi^j_2,\phi^j_{12},p^0,q^0,p^j,q^j)\mapsto (\phi^j,\phi^j_2,\phi^j_1,\phi^j_{12},q^0,p^0,q^j,p^j),\\
\sigma : (\phi^j,\phi^j_1,\phi^j_2,\phi^j_{12},p^0,q^0,p^j,q^j)\mapsto (\phi^j_{12},\phi^j_2,\phi^j_1,\phi^j,q^0,p^0,q^j,p^j),\\
\forall j\in \{1,\ldots,N\},
\end{align*}
i.e. it respects the {\it rombic} symmetry;

\item it is an integrable hierarchy of difference systems in vertex variables.
\end{enumerate}
\end{prop}
\begin{proof}
The proof follows similarly to the proof of Proposition \ref{prop:mpot}.
\end{proof}
 In the non-commutative setting,    (\ref{Vertex:pot2}) can be solved only for the corner variable sets $\{\chi^1,\ldots, \chi^N\}$ and  $\{\chi^1_{12},\ldots, \chi^N_{12}\}.$ If we impose the centrality assumptions   (\ref{Lax:eq3.2}), that in terms of the potential functions (\ref{Lax:eq7}) take the form
    \begin{align} \label{Vertex:aeq2}
    \begin{aligned}
    \left( p^N\chi_1^N-p^1\chi^N\right)\cdots \left(p^2\chi_1^2-p^3\chi^2 \right)\left( p^1\chi_1^1-p^2\chi^1\right)=(p^0)^N,\\
    \left( p^N\chi_{12}^N-p^1\chi_2^N\right)\cdots \left(p^2\chi_{12}^2-p^3\chi_2^2 \right)\left( p^1\chi_{12}^1-p^2\chi_2^1\right)=(p^0)^N,\\
    \left( p^N\chi_2^N-p^1\chi^N\right)\cdots \left(p^2\chi_2^2-p^3\chi^2 \right)\left( p^1\chi_2^1-p^2\chi^1\right)=(q^0)^N,\\
    \left( p^N\chi_{12}^N-p^1\chi_1^N\right)\cdots \left(p^2\chi_{12}^2-p^3\chi_1^2 \right)\left( p^1\chi_{12}^1-p^2\chi_1^1\right)=(p^0)^N,
    \end{aligned}
    \end{align}
    then (\ref{Vertex:pot2}) can be solved  for any corner variable set. 

     The point transformation $\chi^i_{m+k,n+l}\mapsto \left(\frac{p^{i-1}}{p^i}\right)^k\left(\frac{q^{i-1}}{q^i}\right)^l\chi^i_{m+k,n+l},\;k,l\in \mathbb{Z}$ that scales off $p^i, q^i,$ $i=1,2,\ldots,N$  in (\ref{Vertex:pot2}), re-written in terms of $x^i,y^i$ by using (\ref{Lax:eq7}), results to a non-admissible Mobi\"us transformation. So, for arbitrary $N$ the system of equations (\ref{Vertex:pot2}) together with the centrality assumptions (\ref{Vertex:aeq2}), can be considered as  a $2N-$parameter extension of non-commutative lattice Schwarzian Gel'fand-Dikii hierarchy and that is apparent in the examples that follow. In the commutative case the lattice-Schwarzian Gel'fand-Dikii hierarchy was introduced in \cite{JamesPhd}.
\begin{example}[$N=2$]
For $N=2,$ we have the Lax matrix
\begin{align*}
L(\Phi_1,\Phi;P,\lambda)=\begin{pmatrix}
p^1&\lambda \left(p^2\chi^2_1-p^1\chi^2\right)\\
p^1\chi^1_1-p^2\chi^1&p^2
\end{pmatrix}
\end{align*}
and it serves as a strong Lax matrix for the integrable difference system in vertex variables (\ref{Vertex:pot2}) with $N=2$, namely
\begin{align} \label{AVertex:21}
\left(p^1\chi^{1}_{12}-p^{2}\chi^{1}_2 \right)\left(q^{2} \chi^{2}_2-q^1 \chi^{2} \right)-\left(q^1\chi^{1}_{12}-q^{2}\chi^{1}_1 \right)\left(p^{2} \chi^{2}_1-p^1 \chi^{2} \right)=0, \\ \label{AVertex:22}
\left(p^2\chi^{2}_{12}-p^{1}\chi^{2}_2 \right)\left(q^{1} \chi^{1}_2-q^2 \chi^{1} \right)-\left(q^2\chi^{2}_{12}-q^{1}\chi^{2}_1 \right)\left(p^{1} \chi^{1}_1-p^2 \chi^{1} \right)=0.
\end{align}
The centrality assumption (\ref{Vertex:aeq2}) that now reads
\begin{align} \label{central:apot}
\begin{aligned}
\left(p^2\chi^2_1-p^1\chi^2\right)\left(p^1\chi^1_1-p^2\chi^1\right)=\left(p^0\right)^2=\left(p^2\chi^2_{12}-p^1\chi^2_2\right)\left(p^1\chi^1_{12}-p^2\chi^1_2\right),\\
\left(q^2\chi^2_2-q^1\chi^2\right)\left(q^1\chi^1_2-q^2\chi^1\right)=\left(q^0\right)^2=\left(q^2\chi^2_{12}-q^1\chi^2_1\right)\left(q^1\chi^1_{12}-q^2\chi^1_2\right)
\end{aligned}
\end{align}
allow us to re-write the difference system in terms of one potential function only. Indeed, from (\ref{AVertex:21}) using (\ref{central:apot}) we eliminate $\chi^2$ and its shifts to obtain:
\begin{align} \label{Vertex:SKdV}
\left( q^0\right)^2\left(p^1\chi_{12}-p^2\chi_2 \right)\left( q^1\chi_2-q^2\chi\right)^{-1}=
\left( p^0\right)^2\left(q^1\chi_{12}-q^2\chi_1 \right)\left( p^1\chi_1-p^2\chi\right)^{-1},
\end{align}
where for simplicity we have denoted $\chi:=\chi^1,$ $\chi_1:=\chi^1_1,$ etc. When $p^1=p^2=1=q^1=q^2$ and in the commutative setting, equation (\ref{Vertex:SKdV})
coincides with the lattice-Schwarzian KdV equation \cite{Nijhoff:1995}, that is also known as the cross-ratio equation or   $Q1^0$ in \cite{ABS}.  For arbitrary $q^i,p^i,\;i=1,2,$  equation (\ref{Vertex:SKdV}) was introduced in \cite{Kouloukas:2015}.
The non-commutative version of the lattice-Schwarzian KdV equation was introduced in \cite{bs:2002N} and it includes two parameters. The extended non-commutative lattice-Schwarzian KdV equation that was introduced in this Section includes six parameters. So (\ref{Vertex:SKdV}) serves as the non-commutative version of a $4-$parameter extension of the lattice-Schwarzian KdV equation.
\end{example}
\begin{example}[$N\geq 3$]
For $N=3,$ we have the Lax matrix
\begin{align*}
L(\Phi_1,\Phi;P,\lambda)=\begin{pmatrix}
p^1&0&\lambda \left(p^3\chi^3_1-p^1\chi^3\right)\\
p^1\chi^1_1-p^2\chi^1&p^2&0\\
0&p^2\chi^2_1-p^3\chi^2&p^3
\end{pmatrix}
\end{align*}
and it serves as a strong Lax matrix for the integrable difference system in vertex variables (\ref{Vertex:pot2}) with $N=3$, namely
\begin{align} \label{AVertex:31}
\left(p^1\chi^{1}_{12}-p^{2}\chi^{1}_2 \right)\left(q^{3} \chi^{3}_2-q^1 \chi^{3} \right)-\left(q^1\chi^{1}_{12}-q^{2}\chi^{1}_1 \right)\left(p^{3} \chi^{3}_1-p^1 \chi^{3} \right)=0, \\ \label{AVertex:32}
\left(p^2\chi^{2}_{12}-p^{3}\chi^{2}_2 \right)\left(q^{1} \chi^{1}_2-q^2 \chi^{1} \right)-\left(q^2\chi^{2}_{12}-q^{3}\chi^{2}_1 \right)\left(p^{1} \chi^{1}_1-p^2 \chi^{1} \right)=0,\\  \label{AVertex:33}
\left(p^3\chi^{3}_{12}-p^{1}\chi^{3}_2 \right)\left(q^{2} \chi^{2}_2-q^3 \chi^{2} \right)-\left(q^3\chi^{3}_{12}-q^{1}\chi^{3}_1 \right)\left(p^{2} \chi^{2}_1-p^3 \chi^{2} \right)=0,
\end{align}
The centrality assumption (\ref{Vertex:aeq2}) that now reads
\begin{align} \label{central:apot3}
\begin{aligned}
\left(p^3\chi^3_1-p^1\chi^3\right)\left(p^2\chi^2_1-p^3\chi^2\right)\left(p^1\chi^1_1-p^2\chi^1\right)=\left(p^0\right)^3,\\
\left(p^3\chi^3_{12}-p^1\chi^3_2\right)\left(p^2\chi^2_{12}-p^3\chi^2_2\right)\left(p^1\chi^1_{12}-p^2\chi^1_2\right)=\left(p^0\right)^3,\\
\left(q^3\chi^3_2-q^1\chi^3\right)\left(q^2\chi^2_3-q^3\chi^2\right)\left(q^1\chi^1_3-q^2\chi^1\right)=\left(q^0\right)^3,\\
\left(q^3\chi^3_{12}-q^1\chi^3_1\right)\left(q^2\chi^2_{12}-q^3\chi^2_1\right)\left(q^1\chi^1_{12}-q^2\chi^1_1\right)=\left(q^0\right)^3,\\
\end{aligned}
\end{align}
allow us to re-write the difference system in terms of  potential functions $\chi^1$ and $\chi^2$ only. Indeed, by using (\ref{central:apot3}), we eliminate $\chi^3$ and its shifts from  (\ref{AVertex:31}) and (\ref{AVertex:33}) to obtain:
\begin{align} \label{Vertex:Sbuss}
\begin{aligned}
\left( q^0\right)^3\left(p^1\chi^1_{12}-p^2\chi^1_2 \right)\left( q^1\chi^1_2-q^2\chi^1\right)^{-1}\left( q^2\chi^2_2-q^3\chi^2\right)^{-1} & &\\
&\hspace{-4.4cm}=\left( p^0\right)^3\left(q^1\chi^1_{12}-q^2\chi^1_1 \right)\left( p^1\chi^1_1-p^2\chi^1\right)^{-1}\left( p^2\chi^2_1-p^3\chi^2\right)^{-1},&\\
\left( p^0\right)^3\left(p^1\chi^1_{12}-p^2\chi^1_2 \right)^{-1}\left( p^2\chi^2_{12}-p^3\chi^2_2\right)^{-1}\left( q^2\chi^2_2-q^3\chi^2\right) & &\\
&\hspace{-4.4cm}=\left( q^0\right)^3\left(q^1\chi^1_{12}-q^2\chi^1_1 \right)^{-1}\left( q^2\chi^2_{12}-q^3\chi^2_1\right)^{-1}\left( p^2\chi^2_1-p^3\chi^2\right).&
\end{aligned}
\end{align}
 When $p^1=p^2=1=q^1=q^2$ and in the commutative setting, equation (\ref{Vertex:Sbuss})
coincides with the lattice-Schwarzian Boussinesq equation \cite{Ni1}.  So (\ref{Vertex:Sbuss}) stands for the non-commutative 4-parameter extension of the lattice-Schwarzian Boussinesq equation.

For arbitrary $N,$  the hierarchy of difference systems in vertex variables (\ref{Vertex:pot2}) is a non-commutative integrable  hierarchy that when the centrality assumptions (\ref{Vertex:aeq2}) are imposed, it can be considered as   the non-commutative version of a $2N-$parameter extension of the lattice-Schwarzian Gel'fand-Dikii hierarchy. In the commutative setting the lattice-Schwarzian Gel'fand-Dikii hierarchy was introduced in \cite{JamesPhd}. Moreover it includes two parameters.  The extended non-commutative lattice-Schwarzian Gel'fand-Dikii hierarchy that was introduced in this Section includes $2N+2$ parameters.
\end{example}

\section{The Lax matrix $L^{N,2}$ and integrable hierarchies of difference systems} \label{Section:4}
For $k=2,\; N\geq 3$ we have the the Lax matrix
$$
L^{N,2}(X;P,\lambda):=P+\nabla^1 X^{(1)}+\nabla^2 X^{(2)}+\lambda\left( \Delta^1 X^{(1)}+\Delta^2 X^{(2)}\right),
$$
that explicitly reads:
\begin{equation}\label{Lax4:eq0}
L^{N,2}(X;P,\lambda)=\begin{pmatrix}
                      p^1&0& &\cdots &0 & \lambda\, x^{2,N-1}&\lambda\, x^{1,N}\\
                      x^{1,1}&p^2& 0&  &\cdots &0 &\lambda\, x^{2,N} \\
                      x^{2,1}  &x^{1,2}&\ddots &\ddots &   & &0\\
                      0  &\ddots &\ddots&  &  &{} & \vdots\\
                      & &  & & \ddots &\ddots &\\
                      \vdots &{} & &\ddots &\ddots  &p^{N-1} &0\\
                      0     & 0 & {}  & &x^{2,N-2}&  x^{1,N-1}    &p^N
                     \end{pmatrix},
\end{equation}
 The compatibility conditions (\ref{eq0})-(\ref{eq2}) read:
\begin{align}
\begin{aligned}\label{Lax4:eq1}
x^{2,i}_2 y^{2,i+1}=y^{2,i}_1 x^{2,i+1},
\end{aligned}\\
\begin{aligned}\label{Lax4:eq2}
q^i x^{2,i}_2-q^{i-1}x^{2,i}+x^{1,i+1}_2y^{1,i}=p^i y^{2,i}_1-p^{i-1}y^{2,i}+y^{1,i+1}_1x^{1,i},
\end{aligned}\\
\begin{aligned}\label{Lax4:eq3}
x^{2,i}_2y^{1,i-1}+x^{1,i+1}_2y^{2,i-1}=y^{2,i}_1x^{1,i-1}+y^{1,i+1}_1x^{2,i-1},
\end{aligned}\\
\begin{aligned}\label{Lax4:eq4}
q^ix^{1,i}_2-q^{i+1}x^{1,i}=p^iy^{1,i}_1-p^{i+1}y^{1,i},
\end{aligned}
\end{align}
where the superscript $i=1,2,\ldots, N,$ is considered modulo $N$.
\subsection{An integrable hierarchy of difference systems in edge variables}
\begin{prop} \label{Lax4:prop1}
For the hierarchy of difference systems in edge variables  (\ref{Lax4:eq1})-(\ref{Lax4:eq4}), it holds:
\begin{enumerate}
\item matrix (\ref{Lax4:eq0}), serves as its strong Lax matrix;
\item it is birational;
\end{enumerate}
\end{prop}
\begin{proof}
\begin{enumerate}
\item The compatibility conditions (\ref{eq0})-(\ref{eq2}) for the Lax matrix (\ref{Lax4:eq0}) are exactly the hierarchy (\ref{Lax4:eq1})-(\ref{Lax4:eq4}).
\item In matrix form the hierarchy of difference systems    (\ref{Lax4:eq1})-(\ref{Lax4:eq4}) reads: 
\begin{align} \label{Lax4:sol}
\begin{pmatrix}
{\bf x}^1_2&{\bf x}^2_2&{\bf y}^1_1&{\bf y}^2_1
\end{pmatrix}
\begin{pmatrix}
Q& \overline{Y}^{(1)} &\underline{Y}^{(2)} & {\bf 0}_N\\
{\bf 0}_N & \overline{Q} & \underline{Y}^{(1)} & Y^{(2)}\\
-P& -\overline{X}^{(1)} &-\underline{X}^{(2)} & {\bf 0}_N\\
{\bf 0}_N & -\overline{P} & -\underline{X}^{(1)} & -X^{(2)}
\end{pmatrix}=
\begin{pmatrix}
{\bf b}^1 &
{\bf b}^2 &
\bf 0 &
\bf 0
\end{pmatrix},
\end{align}
where ${\bf 0}_N$ stands for the order $N$ zero matrix, while $\bf 0$ stands for the $N-$component zero row vector. In the formulae above participate as well the following $N-$component vectors ${\bf x}^1_2:=(x^{1,1}_2,x^{1,2}_2,\ldots,x^{1,N}_2),$ ${\bf x}^2_2:=(x^{2,N}_2,x^{2,1}_2,\ldots,x^{2,N-1}_2),$ ${\bf y}^1_1:=(y^{1,1}_1,y^{1,2}_1,\ldots,y^{1,N}_1)$ and ${\bf y}^2_1:=(y^{2,N}_1,y^{2,1}_1,\ldots,y^{2,N-1}_1).$ We also have the diagonal order $N$ matrices,
\begin{align*}
&Q:=diag(q^1,q^2,\ldots,q^N),&\overline{Q}:=diag(q^N,q^1,\ldots,q^{N-1}),\\
&X^{(j)}:=diag(x^{j,1},x^{j,2},\ldots, x^{j,N}),& \overline{X}^{(j)}:=diag(x^{j,N},x^{j,1},\ldots,x^{j,N-1}),\\
&\underline{X}^{(j)}:=diag(x^{j,2},\ldots,x^{j,N},x^{j,1}),&{}
 \end{align*}
 $j=1,2$ and similarly are defined the diagonal matrices $P,Y^{(j)},\overline{Y}^{(j)}$ and $\underline{Y}^{(j)}.$ Finally, the $N-$component vectors
  \begin{align*}
  {\bf b}^1:=(q^2x^{1,1}-p^2y^{1,1},q^3x^{1,2}-p^3y^{1,2},\ldots,q^1x^{1,N}-p^1y^{1,N}),\\
   {\bf b}^2:=(q^1x^{2,N}-p^1y^{2,N},q^2x^{2,1}-p^2y^{2,1},\ldots,q^Nx^{2,N-1}-p^Ny^{2,N-1}),
  \end{align*}
also participate in (\ref{Lax4:sol}).

The Schur block matrix inversion formula, states  that for the block matrix
\begin{align*}
M=\begin{pmatrix}
A&B\\
C&D
\end{pmatrix},
\end{align*}
 its inverse matrix $M^{-1}$ reads
\begin{align*}
M^{-1}=\begin{pmatrix}
(M/D)^{-1}&(M/B)^{-1}\\
(M/C)^{-1}&(M/A)^{-1}
\end{pmatrix},&&\begin{array}{ll}
                 M/A:=D-CA^{-1}B,& M/B:=C-DB^{-1}A,\\
                 M/C:=B-AC^{-1}D,& M/D:=A-BD^{-1}C.
                 \end{array}
\end{align*}
By using this inversion formula
in (\ref{Lax4:sol}), we  obtain (\ref{Lax4:eq1})-(\ref{Lax4:eq4}) in  solved  form. Working  similar, we can obtain the inverse solved form of (\ref{Lax4:eq1})-(\ref{Lax4:eq4}) and that proves bi-rationality of the latter.
\end{enumerate}
\end{proof}
The following remarks are in order.
First,  by setting
 \begin{align*}
 x^{2,i}=y^{2,i}=x_2^{2,i}=y_1^{2,i}=0,&& \forall i\in \{1,\ldots,N\},
  \end{align*}
  equations (\ref{Lax4:eq1}), (\ref{Lax4:eq3}) of the hierarchy of difference systems in Proposition \ref{Lax4:prop1} vanish, while equations (\ref{Lax4:eq2}), (\ref{Lax4:eq4}) respectively read
\begin{align} \label{Lax4:aeq5}
x^{1,i+1}_2 y^{1,i}=y^{1,i+1}_1 x^{1,i},\quad q^i x^{1,i}_2-q^{i+1}x^{1,i}=p^i y^{1,i}_1-p^{i+1}y^{1,i},
\end{align}
which coincide respectively with (\ref{Lax:eq2}) and (\ref{Lax:eq1}). So we have obtained a reduction from the hierarchy of difference systems (\ref{Lax4:eq1})-(\ref{Lax4:eq4}) to the hierarchy of difference systems  (\ref{Lax:eq1}),(\ref{Lax:eq2}).
Similarly by setting
\begin{align*}
x^{1,i}=y^{1,i}=x_2^{1,i}=y_1^{1,i}=0,& &\forall i\in \{1,\ldots,N\},
 \end{align*}
 equations (\ref{Lax4:eq3}), (\ref{Lax4:eq4}) vanish while equations (\ref{Lax4:eq1}), (\ref{Lax4:eq2}) respectively read
\begin{align} \label{Lax4:eq5}
x^{2,i}_2 y^{2,i+1}=y^{2,i}_1 x^{2,i+1},\quad q^i x^{2,i}_2-q^{i-1}x^{2,i}=p^i y^{2,i}_1-p^{i-1}y^{2,i},
\end{align}
$i=1,\ldots,N.$ Note that (\ref{Lax4:eq5}) is mapped to (\ref{Lax4:aeq5}) via  $(x^{2,i}_2,y^{2,i}_1,x^{2,i},y^{2,i})\mapsto (x^{1,i}_2,y^{1,i}_1,x^{1,i-2},y^{1,i-2})$.

A second remark is that for the hierarchy of difference systems (\ref{Lax4:eq1})-(\ref{Lax4:eq4}) that is introduced in Proposition \ref{Lax4:prop1} we have $i=1,\ldots,N\geq~3.$ Nevertheless this  system makes perfectly sense even for $N=2$ and although we do not have yet a  Lax matrix  for this case,  we anticipate integrability.

Finally, in the commutative setting it holds $\prod_{i=1}^N x^{2,i}_2=\prod_{i=1}^N x^{2,i}$  and $\prod_{i=1}^N y^{2,i}_1=\prod_{i=1}^N y^{2,i}$  that leads to
\begin{align} \label{Lax4:eq6}
\prod_{i=1}^N x^{2,i}=(p^0)^N=\prod_{i=1}^N x_{2}^{2,i},\quad \prod_{i=1}^N y^{2,i}= (q^0)^N=\prod_{i=1}^N y^{2,i}_1,
\end{align}
where $p^0,$ respectively $q^0,$ are considered functions of $m,$ respectively $n,$ only. Nevertheless, considering (\ref{Lax4:eq6}) together with (\ref{Lax4:eq1})-(\ref{Lax4:eq2}), quadrirationality is not assured as it was the case for the hierarchy of difference systems  (\ref{Lax:eq1}),(\ref{Lax:eq2}). Further restrictions on the dependent variables are required for quadrirationality to be achieved.
\subsection{An integrable hierarchy of difference systems in vertex variables}
The hierarchy of  difference systems of Proposition \ref{Lax4:prop1} constitutes of the sets of equations (\ref{Lax4:eq1})-(\ref{Lax4:eq4}).
The sets of equations (\ref{Lax4:eq1}) and (\ref{Lax4:eq4}) are identically satisfied if we respectively set
\begin{align} \label{Lax5:pot12}
\begin{aligned}
x^{2,i}=p^0 \phi^i_1(\phi^{i+1})^{-1},\quad y^{2,i}=q^0 \phi^i_2(\phi^{i+1})^{-1}, \\
x^{1,i}=p^i\chi^i_1-p^{i+1}\chi^i,\quad y^{1,i}=q^i\chi^i_2-q^{i+1}\chi^i,
\end{aligned} &\quad i=1,\ldots,N.
\end{align}
The functions $\phi^i,\chi^i,$ $i=1,\ldots,N$ can be considered as potential functions. In terms of these potential functions equations (\ref{Lax4:eq1}) and (\ref{Lax4:eq4}) are identically satisfied, while equations (\ref{Lax4:eq2}) and (\ref{Lax4:eq3}) respectively read:
\begin{flalign}
\begin{aligned} \label{Lax5:eq1}
&p^0\left(q^i\phi_{12}^i\left(\phi_2^{i+1}\right)^{-1}-q^{i-1}\phi_1^i\left(\phi^{i+1}\right)^{-1}\right)-
q^0\left(p^i\phi_{12}^i\left(\phi_1^{i+1}\right)^{-1}-p^{i-1}\phi_2^i\left(\phi^{i+1}\right)^{-1}\right)& & \\
&\hspace{0.4cm}=\left(q^{i+1}\chi_{12}^{i+1}-q^{i+2}\chi_1^{i+1}\right)\left(p^i\chi_1^i-p^{i+1}\chi^i\right)-
\left(p^{i+1}\chi_{12}^{i+1}-p^{i+2}\chi_2^{i+1}\right)\left(q^i\chi_2^i-q^{i+1}\chi^i\right),& &
\end{aligned}\\ \vspace{3mm}
\begin{aligned} \label{Lax5:eq2}
&p^0\phi_{12}^i\left(\phi_2^{i+1}\right)^{-1}\left(q^{i-1}\chi_2^{i-1}-q^i\chi^{i-1}\right)+q^0\left(p^{i+1}\chi_{12}^{i+1}-p^{i+2}\chi_2^{i+1}\right)
\phi_2^{i-1}\left(\phi^i\right)^{-1}&&\\
&\hspace{0.4cm}=q^0\phi_{12}^i\left(\phi_1^{i+1}\right)^{-1}\left(p^{i-1}\chi_1^{i-1}-p^i\chi^{i-1}\right)+p^0\left(q^{i+1}\chi_{12}^{i+1}-q^{i+2}\chi_1^{i+1}\right)
\phi_1^{i-1}\left(\phi^i\right)^{-1},& &
\end{aligned}
\end{flalign}
$i=1,\ldots, N$ and constitute a hierarchy of difference systems in vertex variables.
\begin{prop} \label{prop:mpot12}
For the hierarchy  of difference equations in vertex variables (\ref{Lax5:eq1}),(\ref{Lax5:eq2}) it holds
\begin{enumerate}
\item arises as the compatibility conditions of the Lax equation (\ref{lax-eq-ver}), associated with the strong Lax matrix
\begin{align*}
L= \begin{pmatrix}
                      p^1&0& &\cdots &0 & \lambda\, x^{2,N-1}&\lambda\, x^{1,N}\\
                      x^{1,1}&p^2& 0&  &\cdots &0 &\lambda\, x^{2,N} \\
                      x^{2,1}  &x^{1,2}&\ddots &\ddots &   & &0\\
                      0  &\ddots &\ddots&  &  &{} & \vdots\\
                      & &  & & \ddots &\ddots &\\
                      \vdots &{} & &\ddots &\ddots  &p^{N-1} &0\\
                      0     & 0 & {}  & &x^{2,N-2}&  x^{1,N-1}    &p^N
                     \end{pmatrix},
\end{align*}
where   $x^{2,i}:=p^0 \phi^i_1(\phi^{i+1})^{-1},$ $y^{2,i}:=q^0 \phi^i_2(\phi^{i+1})^{-1},$
$x^{1,i}:=p^i\chi^i_1-p^{i+1}\chi^i,$  $y^{1,i}:=q^i\chi^i_2-q^{i+1}\chi^i,$ $i=1,\ldots,N$;
\item it is invariant under the following permutations of the dependent variables
\begin{align*}
\tau :(\phi^j,\phi^j_1,\phi^j_2,\phi^j_{12},\chi^j,\chi^j_1,\chi^j_2,\chi^j_{12},p^0,q^0,p^j,q^j)\mapsto (\phi^j,\phi^j_2,\phi^j_1,\phi^j_{12},\chi^j,\chi^j_2,\chi^j_1,\chi^j_{12},q^0,p^0,q^j,p^j),\\
\sigma : (\phi^j,\phi^j_1,\phi^j_2,\phi^j_{12},\chi^j,\chi^j_1,\chi^j_2,\chi^j_{12},p^0,q^0,p^j,q^j)\mapsto (\phi^j_{12},\phi^j_2,\phi^j_1,\phi^j,\chi^j_{12},\chi^j_2,\chi^j_1,\chi^j,q^0,p^0,q^j,p^j),\\
\forall j\in \{1,\ldots,N\},
\end{align*}
i.e. it respects the {\it rombic} symmetry;
\item it is an integrable hierarchy  of difference equations in vertex variables.
\end{enumerate}
\end{prop}
\begin{proof}
The proof follows similarly to the proof of Proposition \ref{prop:mpot}.
\end{proof}
The following remarks are in order.
First, By setting
\begin{align*}
\chi^i=\chi^i_1=\chi^i_2=\chi^i_{12}=0,&& \forall i\in \{1,\ldots,N\},
\end{align*}
 the set of equations (\ref{Lax5:eq2}) vanishes while (\ref{Lax5:eq1}) reads:
\begin{align*}
p^0\left(q^i\phi_{12}^i\left(\phi_2^{i+1}\right)^{-1}-q^{i-1}\phi_1^i\left(\phi^{i+1}\right)^{-1}\right)=
q^0\left(p^i\phi_{12}^i\left(\phi_1^{i+1}\right)^{-1}-p^{i-1}\phi_2^i\left(\phi^{i+1}\right)^{-1}\right),
\end{align*}
$i=1,\ldots, N,$ that is mapped to (\ref{eq:pot1}) via $(\phi^i,\phi_1^{i+1},\phi_2^{i+1},\phi_{12}^i)\mapsto (\phi^i,\phi_1^{i-1},\phi_2^{i-1},\phi_{12}^i),$  $\forall i\in \{1,\ldots,N\}.$ Furthermore, if the centrality assumption (\ref{Vertex:eq2}) are imposed, we arrive to the non-commutative lattice-modified Gel'fand-Dikii hierarchy.  To recapitulate, we have obtained  a reduction of the hierarchy (\ref{Lax5:eq1}),(\ref{Lax5:eq2}) to the non-commutative lattice-modified Gel'fand-Dikii hierarchy.
Second, by considering $p^0=q^0=0$
 to the sets of equations (\ref{Lax5:eq1}),(\ref{Lax5:eq2}),  equations (\ref{Lax5:eq2})  vanishes while equations (\ref{Lax5:eq1}) read:
\begin{align*}
\left(q^{i+1}\chi_{12}^{i+1}-q^{i+2}\chi_1^{i+1}\right)\left(p^i\chi_1^i-p^{i+1}\chi^i\right)=
\left(p^{i+1}\chi_{12}^{i+1}-p^{i+2}\chi_2^{i+1}\right)\left(q^i\chi_2^i-q^{i+1}\chi^i\right),
\end{align*}
$i=1,\ldots, N$ that coincide with  (\ref{Vertex:pot2}). Furthermore, if the centrality assumptions (\ref{Vertex:aeq2}) are imposed, we arrive to the non-commutative lattice-Schwarzian Gel'fand-Dikii hierarchy. To recapitulate, we have obtained  a reduction of the hierarchy (\ref{Lax5:eq1}),(\ref{Lax5:eq2}) to the non-commutative lattice-Schwarzian Gel'fand-Dikii hierarchy.

\section{Conclusions} \label{Section:5}
In this article we have introduced a family of Lax matrices $L^{N,k}$   that participate to the  linear problem:
\begin{align} \label{Con:1}
\Psi_2=L^{N,k}(X;P,\lambda)\Psi,\quad \Psi_1=L^{N,k}(Y;Q,\lambda)\Psi.
\end{align}
For $k=1,2$ we derived the corresponding integrable hierarchies of difference systems in non-commuting edge variables and the associated integrable difference hierarchies in vertex variables together with their Lax matrices. The lattice-modified Gel'fand-Dikki hierarchy and the lattice-Schwarzian Gel'fand-Dikki hierarchy, both in non-commuting variables, together with the underlying integrable difference system in edge variables, were obtained for $k=1$. For $k=2,$  we have obtained a seemingly novel hierarchy of  difference systems in edge and vertex variables, that includes both the lattice-modified   and the lattice-Schwarzian Gel'fand-Dikki hierarchies, since the latter are obtained by appropriate reductions.
This contribution clearly raises many new questions to be addressed. Let us conclude this article by mentioning a few of them.

  A first question concerns the continuous limit of the discrete hierarchy (\ref{Lax5:eq1}),(\ref{Lax5:eq2}). We do not know yet the continuous hierarchy of equations that corresponds to (\ref{Lax5:eq1}),(\ref{Lax5:eq2}). Furthermore, the discrete hierarchies obtained from (\ref{Con:1}) for $k>2$ await to be studied and their continuous counterparts  to be identified. For example, when $k=3$ the hierarchy of difference systems in edge variables associated with $L^{N,3}$ is given implicitly by equations (\ref{eq0})-(\ref{eq2}) (with $k=3$) and explicitly reads:
  \begin{align} \label{concl:k=3}
  \begin{aligned}
  x_2^{3,i}y^{3,i+1}=  y_1^{3,i}x^{3,i+1},\\
  q^ix_2^{1,i}-q^{i+1}x^{1,i}=  p^i y_1^{1,i}-p^{i+1}y^{1,i},\\
  q^ix_2^{3,i}-q^{i+3}x^{3,i}+x_2^{2,i+1}y^{1,i}+x_2^{1,i+2}y^{2,i}=  p^iy_1^{3,i}-p^{i+3}y^{3,i}+y_1^{2,i+1}x^{1,i}+y_1^{1,i+2}x^{2,i},\\
  q^ix_2^{2,i}-q^{i+2}x^{2,i}+x_2^{1,i+1}y^{1,i}= p^iy_1^{2,i}-p^{i+2}y^{2,i}+y_1^{1,i+1}x^{1,i},\\
  x_2^{3,i}y^{1,i+3}+x_2^{2,i+1}y^{2,i+3}+x_2^{1,i+2}y^{3,i+3}=  y_1^{3,i}x^{1,i+3}+y_1^{2,i+1}x^{2,i+3}+y_1^{1,i+2}x^{3,i+3},\\
  x_2^{3,i}y^{2,i+2}+x_2^{2,i+1}y^{3,i+2}=  y_1^{3,i}x^{2,i+2}+y_1^{2,i+1}x^{3,i+2},
  \end{aligned}
  \end{align}
$i=1,\ldots, N.$ The first two sets of equations of (\ref{concl:k=3}) are identically satisfied by setting
\begin{align*}
\begin{aligned}
x^{3,i}=p^0\phi_1^i\left(\phi^{i+1}\right)^{-1},&& y^{3,i}=q^0\phi_2^i\left(\phi^{i+1}\right)^{-1},\\
x^{1,i}=p^i\chi_1^i-p^{i+1}\chi^i,  && y^{1,i}=q^i\chi_2^i-q^{i+1}\chi^i,
\end{aligned} && i=1,\ldots, N,
\end{align*}
and in terms of the potential functions $\phi^i,\chi^i,$ $i=1,\ldots, N,$ the remaining sets of equations of (\ref{concl:k=3}) constitute a hierarchy of difference systems in edge and vertex variables.

   For the Lax matrix $L^{N,1}$ we have implicitly presented the hierarchy of associated Yang-Baxter maps (see Proposition \ref{prop1.2}). The fist member of this hierarchy of  maps is a $4-$parametric extension of the $H_{III}^A$ Yang-Baxter map, extended  in the non-commutative domain and it is  presented in Example \ref{example:1}, whereas the second member of this hierarchy is explicitly presented in Example \ref{example:1.2}. We anticipate to present the explicit form of the whole hierarchy of Yang-Baxter maps elsewhere as well as the hierarchy of {\em entwining maps} \cite{Kassotakis:2019} associated with this hierarchy.    Furthermore, an open question is to obtain explicitly the hierarchies of Yang-Baxter maps associated with the Lax matrices $L^{N,k}$ for $k\geq 2.$

    The linear problem (\ref{Con:1}) falls into a class of more general linear problems, namely
\begin{align} \label{Con:1.1}
\Psi_2=L^{N,k_1,a}(X;P,\lambda)\Psi,\quad \Psi_1=M^{N,k_2,b}(Y;Q,\lambda)\Psi,
\end{align}
where
\begin{align} \label{Con:2}
\begin{aligned}
L^{N,k_1,a}(X;P,\lambda):=\sum_{i=0}^{k_1}\nabla^{i}X^{(i)}+\lambda\sum_{i=0}^{k_1}\Delta^{i}X^{(i)},\\
M^{N,k_2,b}(Y;Q,\lambda):=\sum_{i=0}^{k_2}\nabla^{i}Y^{(i)}+\lambda\sum_{i=0}^{k_2}\Delta^{i}Y^{(i)},
\end{aligned}
\end{align}
with $N\in \mathbb{N},$ $k_1,k_2,a,b\in\{0,1,\ldots,N-1\},$ $k_1\geq a,$ $k_2\geq b$ and $\nabla^0=\Delta^0\equiv I_N,$ where $I_N$ the order $N$ identity matrix. Also $X^{(a)}\equiv P$ the diagonal order $N$ matrix with entries the parameters $(X^{(a)})_{i,i}:=p^i$;   $Y^{(a)}\equiv Q$ the diagonal order $N$ matrix with entries the parameters $(Y^{(a)})_{i,i}:=q^i$ and $X^{(j)},$ $Y^{(j)},$ $j=0,1,\ldots N-1, j\neq a,b,$ the order $N$ diagonal matrices defined on Section \ref{secdef:lax}. Note that (\ref{Con:1}) arises as the special case of (\ref{Con:1.1}),(\ref{Con:2}), when $k_1=k_2=k,$ $a=b=0.$  A classification of the linear problems (\ref{Con:1.1}),(\ref{Con:2}),  awaits to be addressed.

Finally,  in the seminal articles \cite{Doliwa_2013,Doliwa_2014}, the geometric interpretation of the KP map in terms of   Desargues maps is provided. The geometric interpretation of the difference systems in edge variables that correspond to Lax matrices $L^{N,k}$ with $k\geq 2$ and  when we assume that $N\in \mathbb{Z},$ is an open question.

 \subsubsection*{{\bf Aknowledgements}} The author would like to thank Theodoros Kouloukas and Maciej Nieszporski for fruitful discussions.


\begin{thebibliography}{10}

\bibitem{Gel:1976}
I.M. Gel'fand and L.A. Dikii.
\newblock Fractional powers of operators and {H}amiltonian systems.
\newblock {\em Funct. Anal. Appl.}, 10(4):259--273, 1976.

\bibitem{Manin:1979}
Yu.I. Manin.
\newblock Algebraic aspects of nonlinear differential equations.
\newblock {\em J. Sov. Math.}, 11:1--122, 1979.

\bibitem{Drinfeld:1981}
V.G. Drinfeld and V.V. Sokolov.
\newblock Equations of {K}orteweg–-de {V}ries type, and simple {L}ie algebras.
\newblock {\em Dokl. Akad. Nauk SSSR}, 258(1):1--122, 1981.

\bibitem{Mikhailov:1979}
A.V. Mikhailov.
\newblock Integrability of a two-dimensional generalization of the toda chain.
\newblock {\em Pis'ma Zh. Eksp. Teor. Fiz.}, 30(7):443--448, 1979.

\bibitem{Fordy:1980}
A.P. Fordy and J.~Gibbons.
\newblock Integrable nonlinear {K}lein-{G}ordon equations and {T}oda lattices.
\newblock {\em Commun. Math. Phys.}, 77:21--30, 1980.

\bibitem{Fordy:1993}
M.~Antonowicz and A.P. Fordy.
\newblock Multicomponent {S}chwarzian {K}d{V} hierarchies.
\newblock {\em Rep. Mod. Phys.}, 32:223--233, 1993.

\bibitem{Adler:2021}
V.E. Adler and V.V. Sokolov.
\newblock Non-{A}belian evolution systems with conservation laws.
\newblock {\em Math. Phys. Anal. Geom.}, 24:7, 2021.

\bibitem{Nijhoff:1996}
F.W. Nijhoff and V.G. Papageorgiou.
\newblock On some integrable discrete-time systems associated with the
  {B}ogoyavlensky lattices.
\newblock {\em Physica A}, 228:172--–188, 1996.

\bibitem{Tongas:2004}
A.~Tongas and F.~Nijhoff.
\newblock The {B}oussinesq integrable system: compatible lattice and continuum
  structures.
\newblock {\em Glasgow Math. J.}, 47A:205–--219, 2004.

\bibitem{Maruno:2010}
K.~Maruno and K.~Kajiwara.
\newblock The discrete potential {B}oussinesq equation and its multisoliton
  solutions.
\newblock {\em Appl. Anal.}, 89:593–--609, 2010.

\bibitem{Hietarinta:2011}
J.~Hietarinta.
\newblock Boussinesq-like multi-component lattice equations and
  multi-dimensional consistency.
\newblock {\em J. Phys. A: Math. Theor.}, 44:165204, 2011.

\bibitem{JamesPhd}
J.~Atkinson.
\newblock {\em Integrable lattice equations: connection to the M\"obius group,
  B\"acklund transformations and solutions}.
\newblock PhD thesis, University of Leeds, 2008.
\newblock http://etheses.whiterose.ac.uk/9081/.

\bibitem{Hay:2014}
C.~Scimiterna, M.~Hay, and D.~Levi.
\newblock On the integrability of a new lattice equation found by multiple
  scale analysis.
\newblock {\em J. Phys. A: Math. Theor.}, 47:265204, 2014.

\bibitem{Mikhailov:2016}
A.V. Mikhailov, G.~Papamikos, and J.P. Wang.
\newblock Darboux transformation for the vector sine-{G}ordon equation and
  integrable equations on a sphere.
\newblock {\em Lett. Math. Phys.}, 106:973–--996, 2016.

\bibitem{Nalini:2018}
Nolan~M. Joshi~N., Lobb~S.
\newblock Constructing initial value spaces of lattice equations.
\newblock {\em arXiv:1807.06162[nlin]}, 2018.

\bibitem{Kels:2019}
A.P. Kels.
\newblock Extended {Z}-invariance for integrable vector and face models and
  multi-component integrable quad equations.
\newblock {\em J. Stat. Phys.}, 176:1375–--1408, 2019.

\bibitem{Kels:2019II}
A.P. Kels.
\newblock Two-component {Y}ang-{B}axter maps associated to integrable quad
  equations.
\newblock {\em arXiv:1910.03562v5 [math-ph]}, 2019.

\bibitem{Kass2}
P.~Kassotakis, M.~Nieszporski, V.~Papageorgiou, and A.~Tongas.
\newblock Integrable two-component difference systems of equations.
\newblock {\em Proc. R. Soc. A.}, 476:20190668, 2020.

\bibitem{Kamp:2020}
D.~Zhang, P.H van~der Kamp, and D.-J. Zhang.
\newblock Multi-component extension of {CAC} systems.
\newblock {\em SIGMA}, 16(060):30pages, 2020.

\bibitem{Hietarinta:2020}
J.~Hietarinta and D.-J. Zhang.
\newblock Discrete {B}oussinesq--type equations.
\newblock {\em arXiv:2012.00495[nlin.SI]}, 2020.

\bibitem{Nijhoff:1992}
F.W. Nijhoff, V.G. Papageorgiou, H.W. Capel, and G.R.W. Quispel.
\newblock The lattice {G}el'fand-{D}ikii hierarchy.
\newblock {\em Inverse Problems}, 8(4):597--621, aug 1992.

\bibitem{Atkinson:2012}
J.~Atkinson, S.B. Lobb, and F.W. Nijhoff.
\newblock An integrable multicomponent quad-equation and its {L}agrangian
  formulation.
\newblock {\em Theor. Math. Phys.}, 173:1644–--1653, 2012.

\bibitem{Ni1}
F.W. Nijhoff.
\newblock On some ``{S}chwarzian equations'' and their discrete analogues.
\newblock In A.S. Fokas and I.M. Gel'fand, editors, {\em Algebraic Aspects of
  Integrable Systems: In memory of {I}rene {D}orfman}, pages 237--260.
  Birkhäuser Verlag, 1996.

\bibitem{Doliwa_2013}
A.~Doliwa.
\newblock Non-commutative lattice-modified {G}el'fand-{D}ikii systems.
\newblock {\em J. Phys. A: Math. Theor.}, 46(20):205202, 2013.

\bibitem{Doliwa_2014}
A.~Doliwa.
\newblock Non-commutative rational {Y}ang–-{B}axter maps.
\newblock {\em Lett. Math. Phys.}, 104:299–--309, 2014.

\bibitem{Nijhoff:1990}
F.W. Nijhoff and H.W. Capel.
\newblock The direct linearization approach to hierarchies of integrable
  {P}{D}{E}s in $2+1$ dimensions: {I}. {L}attice equations and the
  differential-difference hierarchies.
\newblock {\em Inverse Problems}, 6:567–--590, 1990.

\bibitem{Boris:2000}
B.~Kupershmidt.
\newblock {\em K{P} or m{K}{P}: {N}oncommutative Mathematics of {L}agrangian,
  {H}amiltonian, and Integrable Systems}.
\newblock AMS, Providence, 2000.

\bibitem{Dimakis:2002}
A~Dimakis and F.~M\"{u}ller-Hoissen.
\newblock On generalized {L}otka-{V}olterra lattices.
\newblock {\em Czech. J. Phys.}, 52:1187--1193, 2002.

\bibitem{bs:2002N}
A.~I. Bobenko and Yu.~B. Suris.
\newblock Integrable noncommutative equations on quad-graphs. the consistency
  approach.
\newblock {\em Lett. Math. Phys.}, 61(3):241--254, 2002.

\bibitem{Field:2005}
C.M. Field, F.W. Nijhoff, and H.W. Capel.
\newblock Exact solutions of quantum mappings from the lattice {KdV} as
  multi-dimensional operator difference equations.
\newblock {\em J. Phys. A: Math. Gen.}, 38(43):9503--9527, 2005.

\bibitem{Nimmo:2006}
J.J.C Nimmo.
\newblock On a non-{A}belian {H}irota-{M}iwa equation.
\newblock {\em J. Phys. A: Math. Gen.}, 39:5053–--5065, 2006.

\bibitem{Doliwa_Painleve_2013}
A.~Doliwa.
\newblock Non-commutative q-{P}ainlev{\'{e}} {VI} equation.
\newblock {\em J. Phys. A: Math. Theor.}, 47(3):035203, 2013.

\bibitem{Grahovski:2016}
G.G. Grahovski, S.~Konstantinou-Rizos, and A.V. Mikhailov.
\newblock Grassmann extensions of {Y}ang{\textendash}{B}axter maps.
\newblock {\em J. Phys. A: Math. Theor.}, 49(14):145202, 2016.

\bibitem{Rizos:2016}
S.~Konstantinou-Rizos and A.V. Mikhailov.
\newblock Anticommutative extension of the {A}dler map.
\newblock {\em J. Phys. A: Math. Theor.}, 49(30):30LT03, 2016.

\bibitem{Rizos:20182}
S.~Konstantinou-Rizos and T.E. Kouloukas.
\newblock A noncommutative discrete potential {K}d{V} lift.
\newblock {\em J. Math. Phys.}, 59:063506, 2018.

\bibitem{Kass1}
P.~Kassotakis, M.~Nieszporski, V.~Papageorgiou, and A.~Tongas.
\newblock Tetrahedron maps and symmetries of three dimensional integrable
  discrete equations.
\newblock {\em J. Math. Phys.}, 60:123503, 2019.

\bibitem{Noumi:2020}
A.~Doliwa and M.~Noumi.
\newblock The {C}oxeter relations and {KP} map for non-commuting symbols.
\newblock {\em Lett. Math. Phys.}, 110:2743–--2762, 2020.

\bibitem{ABS:YB}
V.E. Adler, A.I. Bobenko, and Yu.B. Suris.
\newblock Geometry of {Y}ang-{B}axter maps: pencils of conics and
  quadrirational mappings.
\newblock {\em Comm. Anal. Geom.}, 12(5):967--1007, 2004.

\bibitem{Papageorgiou:2010}
V.G. Papageorgiou, Yu.B. Suris, A.G. Tongas, and A.P. Veselov.
\newblock On quadrirational {Y}ang-{B}axter maps.
\newblock {\em SIGMA}, 6:9pp, 2010.

\bibitem{Sklyanin:1988}
E.K. Sklyanin.
\newblock Classical limits of {S}{U}$(2)$--invariant solutions of the
  {Y}ang--{B}axter equation.
\newblock {\em J. Soviet Math.}, 40:93--107, 1988.

\bibitem{Drinfeld:1992}
V.G. Drinfeld.
\newblock On some unsolved problems in quantum group theory, quantum groups.
\newblock {\em Lecture Notes in Math.}, 1510:1--8, 1992.

\bibitem{Bukhshtaber:1998}
V.M. Bukhshtaber.
\newblock Yang–-{B}axter mappings.
\newblock {\em Uspekhi Mat. Nauk}, 53:241--242, 1998.

\bibitem{Veselov:20031}
A.P. Veselov.
\newblock Yang-{B}axter maps and integrable dynamics.
\newblock {\em Phys. Lett. A}, 314:214--221, 2003.

\bibitem{doliwa-santini}
A.~Doliwa and P.M. Santini.
\newblock The symmetric, {D}-invariant and {E}gorov reductions of the
  quadrilateral lattice.
\newblock {\em Journal of Geometry and Physics}, 36(1):60--102, 2000.

\bibitem{Kassotakis_2011}
P.~Kassotakis and M.~Nieszporski.
\newblock Families of integrable equations.
\newblock {\em SIGMA}, 7(100):14pp, 2011.

\bibitem{Kassotakis_2012}
P.~Kassotakis and M.~Nieszporski.
\newblock On non-multiaffine consistent-around-the-cube lattice equations.
\newblock {\em Phys. Lett. A}, 376(45):3135--3140, 2012.
\newblock arXiv:1106.0435.

\bibitem{Kassotakis_2018}
P.~Kassotakis and M.~Nieszporski.
\newblock Difference systems in bond and face variables and non-potential
  versions of discrete integrable systems.
\newblock {\em J. Phys. A: Math. Theor.}, 51(38):385203, 2018.

\bibitem{Fordy_2017}
A.P. Fordy and P.~Xenitidis.
\newblock {$\mathbb{Z}^N$} graded discrete {L}ax pairs and integrable
  difference equations.
\newblock {\em J. Phys. A: Math. Theor.}, 50(16):165205, 2017.

\bibitem{Kassotakis_2021}
M.~Nieszporski and P.~Kassotakis.
\newblock Systems of difference equations on a vector valued function that
  admits 3d vector space of scalar potentials.
\newblock {\em arXiv:1908.01706[nlin]}, 2019.

\bibitem{Papageorgiou:2009II}
V.G Papageorgiou and P.D Xenitidis.
\newblock Symmetries and integrability of discrete equations defined on a
  black-white lattice.
\newblock {\em J. Phys. A: Math. Theor.}, 42:454025, 2009.

\bibitem{ABS:2009}
V.E. Adler, A.I. Bobenko, and Yu.B. Suris.
\newblock Discrete nonlinear hyperbolic equations. {C}lassification of
  integrable cases.
\newblock {\em Funct. Anal. Appl.}, 43(1):3--17, 2009.

\bibitem{Boll:2011}
R.~Boll.
\newblock Classification of 3{D} consistent quad-equations.
\newblock {\em J. Nonlinear Math. Phys.}, 18(3):337--365, 2011.

\bibitem{Bianchi:1894}
L.~Bianchi.
\newblock {\em Lezioni di geometrica differenziale}.
\newblock Enrico Spoerri, 1894.

\bibitem{Hir-sG}
R.~Hirota.
\newblock Nonlinear partial difference equations iii; discrete sine-gordon
  equation.
\newblock {\em J. Phys. Soc. Jpn.}, 43:2079--2086, 1977.

\bibitem{nij-qui-cap}
F.W. Nijhoff, G.R.W. Quispel, and H.W. Capel.
\newblock Direct linearization of nonlinear difference-difference equations.
\newblock {\em Phys. Lett. A}, 97:125--128, 1983.

\bibitem{ABS}
V.E. Adler, A.I. Bobenko, and Yu.B. Suris.
\newblock Classification of integrable equations on quad-graphs. {T}he
  consistency approach.
\newblock {\em Comm. Math. Phys.}, 233(3):513--543, 2003.

\bibitem{Kouloukas:2012}
T.E. Kouloukas and V.G. Papageorgiou.
\newblock 3{D} compatible ternary systems and {Yang-Baxter} maps.
\newblock {\em J. Phys. A: Math. Theor.}, 45(34):345204, 2012.

\bibitem{Nijhoff:1995}
F.~Nijhoff and H.~Capel.
\newblock The discrete {K}orteweg-de {V}ries equation.
\newblock {\em Acta Applicandae Mathematica}, 39(1):133--158, 1995.

\bibitem{Kouloukas:2015}
T.E. Kouloukas and D.T. Tran.
\newblock Poisson structures for lifts and periodic reductions of integrable
  lattice equations.
\newblock {\em J. Phys. A: Math. Theor.}, 48(7):075202, jan 2015.

\bibitem{Kassotakis:2019}
P~Kassotakis.
\newblock Invariants in separated variables: {Y}ang-{B}axter, entwining and
  transfer maps.
\newblock {\em SIGMA}, 15(048):36pp, 2019.

\end{thebibliography}
\end{document}